\documentclass[a4paper,journal]{IEEEtran}
\usepackage{graphicx,ifthen}
\usepackage{psfrag,amssymb,amsthm}
\usepackage[cmex10]{amsmath}
\usepackage{mathtools}

\usepackage{cite,color,pst-plot,stfloats,pst-sigsys}
\usepackage{pstricks-add}
\usepackage{subfigure}

\newtheorem{thm}{Theorem}
\newtheorem{lem}{Lemma}

\newtheorem{prop}{Proposition}
\newtheorem{remark}{Remark}
\theoremstyle{definition}
\newtheorem{definition}{Definition}
\newtheorem{assumption}{Assumption}

\title{Optimal Kullback-Leibler Aggregation via Information Bottleneck}

\author{Bernhard C. Geiger,~\IEEEmembership{Member, IEEE}, Tatjana Petrov, Gernot Kubin,~\IEEEmembership{Member, IEEE}, Heinz Koeppl%
\thanks{Parts of the work of Bernhard C. Geiger have been funded by the Austrian Research Association, project number 06/12684. Heinz Koeppl acknowledges the support from the Swiss National Science Foundation, grant number PP00P2 128503/1. Tatjana Petrov acknowledges the support by SystemsX.ch (the Swiss Inititative for Systems Biology) and by SNSF - the Swiss National Science Fondation.}
\thanks{Bernhard C. Geiger (geiger@ieee.org) is with the Institute for Communications Engineering, TU Munich, and with the Signal Processing and Speech Communication Laboratory, Graz University of Technology.}
\thanks{Tatjana Petrov is with the Automatic Control Lab, ETH Z\"urich and with the IST Austria}
\thanks{Gernot Kubin (g.kubin@ieee.org) is with the Signal Processing and Speech Communication Laboratory, Graz University of Technology.}
\thanks{Heinz Koeppl is with the Department of Electrical Engineering and Information Technology, TU Darmstadt.}
} 

\def \single {0}

\begin{document}

\newcounter{myTempCnt}

\newcommand{\Yvecg}{\Yvec_g}
\newcommand{\Xvecg}{\Xvec_g}

\newcommand{\x}[1]{x[#1]}
\newcommand{\y}[1]{y[#1]}

\newcommand{\pdfy}{f_Y(y)}

\newcommand{\ent}[1]{H(#1)}
\newcommand{\diffent}[1]{h(#1)}
\newcommand{\derate}[1]{\bar{h}\left(\mathbf{#1}\right)}
\newcommand{\mutinf}[1]{I(#1)}
\newcommand{\ginf}[1]{I_G(#1)}
\newcommand{\kld}[2]{D(#1||#2)}
\newcommand{\kldrate}[2]{\bar{D}(\mathbf{#1}||\mathbf{#2})}
\newcommand{\kldr}[2]{\bar{D}({#1}||{#2})}
\newcommand{\binent}[1]{H_2(#1)}
\newcommand{\binentneg}[1]{H_2^{-1}\left(#1\right)}
\newcommand{\entrate}[1]{\bar{H}(\mathbf{#1})}
\newcommand{\mutrate}[1]{\mutinf{\mathbf{#1}}}
\newcommand{\redrate}[1]{\bar{R}(\mathbf{#1})}
\newcommand{\pinrate}[1]{\vec{I}(\mathbf{#1})}
\newcommand{\loss}[2][\empty]{\ifthenelse{\equal{#1}{\empty}}{L(#2)}{L_{#1}(#2)}}
\newcommand{\lossrate}[2][\empty]{\ifthenelse{\equal{#1}{\empty}}{L(\mathbf{#2})}{L_{\mathbf{#1}}(\mathbf{#2})}}
\newcommand{\gain}[1]{G(#1)}
\newcommand{\atten}[1]{A(#1)}
\newcommand{\relLoss}[2][\empty]{\ifthenelse{\equal{#1}{\empty}}{l(#2)}{l_{#1}(#2)}}
\newcommand{\relLossrate}[1]{l(\mathbf{#1})}
\newcommand{\relTrans}[1]{t(#1)}
\newcommand{\partEnt}[2]{H^{#1}(#2)}

\newcommand{\dom}[1]{\mathcal{#1}}
\newcommand{\indset}[1]{\mathbb{I}\left({#1}\right)}

\newcommand{\unif}[2]{\mathcal{U}\left(#1,#2\right)}
\newcommand{\chis}[1]{\chi^2\left(#1\right)}
\newcommand{\chir}[1]{\chi\left(#1\right)}
\newcommand{\normdist}[2]{\mathcal{N}\left(#1,#2\right)}
\newcommand{\Prob}[1]{\mathrm{Pr}(#1)}
\newcommand{\Mar}[1]{\mathrm{Mar}(#1)}
\newcommand{\Qfunc}[1]{Q\left(#1\right)}

\newcommand{\expec}[1]{\mathrm{E}\left\{#1\right\}}
\newcommand{\expecwrt}[2]{\mathrm{E}_{#1}\left\{#2\right\}}
\newcommand{\var}[1]{\mathrm{Var}\left\{#1\right\}}
\renewcommand{\det}{\mathrm{det}}
\newcommand{\cov}[1]{\mathrm{Cov}\left\{#1\right\}}
\newcommand{\sgn}[1]{\mathrm{sgn}\left(#1\right)}
\newcommand{\sinc}[1]{\mathrm{sinc}\left(#1\right)}
\newcommand{\e}[1]{\mathrm{e}^{#1}}
\newcommand{\multint}{\iint{\cdots}\int}
\newcommand{\modd}[3]{((#1))_{#2}^{#3}}
\newcommand{\quant}[1]{Q\left(#1\right)}
\newcommand{\card}[1]{\mathrm{card}(#1)}
\newcommand{\diam}[1]{\mathrm{diam}(#1)}

\newcommand{\ivec}{\mathbf{i}}
\newcommand{\hvec}{\mathbf{h}}
\newcommand{\gvec}{\mathbf{g}}
\newcommand{\avec}{\mathbf{a}}
\newcommand{\kvec}{\mathbf{k}}
\newcommand{\fvec}{\mathbf{f}}
\newcommand{\vvec}{\mathbf{v}}
\newcommand{\xvec}{\mathbf{x}}
\newcommand{\Xvec}{\mathbf{X}}
\newcommand{\Xhvec}{\hat{\mathbf{X}}}
\newcommand{\xhvec}{\hat{\mathbf{x}}}
\newcommand{\xtvec}{\tilde{\mathbf{x}}}
\newcommand{\Yvec}{\mathbf{Y}}
\newcommand{\yvec}{\mathbf{y}}
\newcommand{\Zvec}{\mathbf{Z}}
\newcommand{\Svec}{\mathbf{S}}
\newcommand{\Nvec}{\mathbf{N}}
\newcommand{\Pvec}{\mathbf{P}}
\newcommand{\muvec}{\boldsymbol{\mu}}
\newcommand{\wvec}{\mathbf{w}}
\newcommand{\Wvec}{\mathbf{W}}
\newcommand{\Hmat}{\mathbf{H}}
\newcommand{\Amat}{\mathbf{A}}
\newcommand{\Fmat}{\mathbf{F}}

\newcommand{\zerovec}{\mathbf{0}}
\newcommand{\eye}{\mathbf{I}}
\newcommand{\evec}{\mathbf{i}}

\newcommand{\zeroone}{\left[\begin{array}{c}\zerovec^T\\ \eye\end{array} \right]}
\newcommand{\zerooneT}{\left[\begin{array}{cc}\zerovec & \eye\end{array} \right]}
\newcommand{\zerooneM}{\left[\begin{array}{cc}\zerovec &\zerovec^T\\\zerovec& \eye\end{array} \right]}

\newcommand{\Cxx}{\mathbf{C}_{XX}}
\newcommand{\Cx}{\mathbf{C}_{\Xvec}}
\newcommand{\Chx}{\hat{\mathbf{C}}_{\Xvec}}
\newcommand{\Cy}{\mathbf{C}_{\Yvec}}
\newcommand{\Cz}{\mathbf{C}_{\Zvec}}
\newcommand{\Cn}{\mathbf{C}_{\mathbf{N}}}
\newcommand{\Cnt}{\underline{\mathbf{C}}_{\tilde{\mathbf{N}}}}
\newcommand{\Cntm}{\underline{\mathbf{C}}_{\tilde{\mathbf{N}}}}
\newcommand{\Cxh}{\mathbf{C}_{\hat{X}\hat{X}}}
\newcommand{\rxx}{\mathbf{r}_{XX}}
\newcommand{\Cxy}{\mathbf{C}_{XY}}
\newcommand{\Cyy}{\mathbf{C}_{YY}}
\newcommand{\Cnn}{\mathbf{C}_{NN}}
\newcommand{\Cyx}{\mathbf{C}_{YX}}
\newcommand{\Cygx}{\mathbf{C}_{Y|X}}
\newcommand{\Wmat}{\underline{\mathbf{W}}}

\newcommand{\Jac}[2]{\mathcal{J}_{#1}(#2)}

\newcommand{\NN}{{N{\times}N}}
\newcommand{\perr}{P_e}
\newcommand{\perh}{\hat{\perr}}
\newcommand{\pert}{\tilde{\perr}}

\newcommand{\vecind}[1]{#1_0^n}
\newcommand{\roots}[2]{{#1}_{#2}^{(i_{#2})}}
\newcommand{\rootx}[1]{x_{#1}^{(i)}}
\newcommand{\rootn}[2]{x_{#1}^{#2,(i)}}

\newcommand{\markkern}[1]{f_M(#1)}
\newcommand{\pole}{a_1}
\newcommand{\preim}[1]{g^{-1}(#1)}
\newcommand{\preimV}[1]{\mathbf{g}^{-1}[#1]}
\newcommand{\Xmax}{\bar{X}}
\newcommand{\Xmin}{\underbar{X}}
\newcommand{\xmax}{x_{\max}}
\newcommand{\xmin}{x_{\min}}
\newcommand{\limn}{\lim_{n\to\infty}}
\newcommand{\limX}{\lim_{\hat{\Xvec}\to\Xvec}}
\newcommand{\limx}{\lim_{\hat{X}\to X}}
\newcommand{\limXo}{\lim_{\hat{X}_1\to X_1}}
\newcommand{\sumin}{\sum_{i=1}^n}
\newcommand{\finv}{f_\mathrm{inv}}
\newcommand{\ejtheta}{\e{\jmath\theta}}
\newcommand{\khat}{\bar{k}}
\newcommand{\modeq}[1]{g(#1)}
\newcommand{\partit}[1]{\mathcal{P}_{#1}}
\newcommand{\psd}[1]{S_{#1}(\e{\jmath \theta})}
\newcommand{\borel}[1]{\mathfrak{B}(#1)}
\newcommand{\infodim}[1]{d(#1)}

\newcommand{\delay}[2]{\psblock(#1){#2}{\footnotesize$z^{-1}$}}
\newcommand{\Quant}[2]{\psblock(#1){#2}{\footnotesize$\quant{\cdot}$}}
\newcommand{\moddev}[2]{\psblock(#1){#2}{\footnotesize$\modeq{\cdot}$}}

\def\unifc{\lambda}
\def\N{{\mathbb N}}

\newcommand{\PP}{\ensuremath{\mathsf{P}}}
\def\rA#1{\stackrel{{#1}}{\ra}}

\def\ra{\rightarrow}

\def\longrightharpoonup{\relbar\joinrel\rightharpoonup}
\def\longleftharpoondown{\leftharpoondown\joinrel\relbar}

\def\longrightleftharpoons{
  \mathop{
    \vcenter{
      \hbox{
    \ooalign{
      \raise1pt\hbox{$\longrightharpoonup\joinrel$}\crcr
      \lower1pt\hbox{$\longleftharpoondown\joinrel$}
    }
      }
    }
  }
}

\def\deltaT{\Delta}
\newcommand\cons[1]{\nu_{#1}}
\newcommand\prdn[1]{\nu'_{#1}}
\newcommand\gaint[1]{\mu_{#1}}

\newcommand\consvec[1]{\boldsymbol{\nu}_{#1}}
\newcommand\prdnvec[1]{\boldsymbol{\nu}'_{#1}}
\newcommand\gaintvec[1]{\boldsymbol{\mu}_{#1}}

\def\N{{\mathbb N}}

\newcommand{\Rc}[1]{{R_{#1}}}

\newcommand\sprop[1]{\lambda_{#1}}

\newcommand\dprop[1]{\tilde{\lambda}_{#1}}

\newcommand{\E}{\ensuremath{\mathsf{E}}}

\newcommand{\Qvec}{\mathbf{Q}}
\newcommand{\nuvec}{\boldsymbol{\nu}}
\newcommand{\Vvec}{\mathbf{V}}
\newcommand{\Uvec}{\mathbf{U}}
\newcommand{\pivec}{\boldsymbol{\pi}}

\newcommand{\ka}[1]{\mathtt{#1}}
\newcommand{\rateR}[1]{{k_{#1}}}
\newcommand{\srate}[1]{{c_{#1}}}

\def\specSpace{\N^{\specN}}

\newcommand{\conc}[1]{[{#1}]}

\newcommand{\dxvec}{z}
\newcommand{\dxvecT}{Z}

\def\R{\mathbb R}
\newcommand{\specStateSpace}{\mathcal{X}}

\newcommand{\concT}[2]{[{#1}]({#2})}

\newcommand{\spec}[1]{S_{#1}}
\def\specN{n}
\def\reacN{r}

\newcommand{\ctp}[1]{ {{\textcolor{orange}{TP: #1}}} }
\newcommand{\cbg}[1]{ {{\textcolor{blue}{BG: #1}}} }

\def\p{p}
\def\ptilde{\tilde{p}}

\maketitle

\begin{abstract}
%
%
In this paper, we present a method for reducing a regular, discrete-time Markov chain (DTMC) to another DTMC with a given, typically much smaller number of states. The cost of reduction is defined as the Kullback-Leibler divergence rate between a projection of the original process through a partition function and a DTMC on the correspondingly partitioned state space. Finding the reduced model with minimal cost is computationally expensive, as it requires an exhaustive search among all state space partitions, and an exact evaluation of the reduction cost for each candidate partition. Our approach deals with the latter problem by minimizing an upper bound on the reduction cost instead of minimizing the exact cost; The proposed upper bound is easy to compute and it is tight if the original chain is \emph{lumpable} with respect to the partition. Then, we express the problem in the form of information bottleneck optimization, and propose using the agglomerative information bottleneck algorithm for searching a sub-
optimal partition greedily, rather than exhaustively. The theory is illustrated with examples and one application scenario in the context of modeling bio-molecular interactions.
\end{abstract}


\section{Introduction}

Markov models are ubiquitously used in scientific and engineering disciplines, for example, to understand chemical reaction systems, to model speech recognition and data sources, or in Markov decision processes in automated control. The popularity of these models arises because the Markov property often renders model analysis tractable and their simulation efficient. However, sometimes the state space of a Markov model (i.e., its alphabet) is too large to permit simulation, even when harnessing today's computing power. 
Indeed, in stochastic modeling in computational biology~\cite{Wilkinson_SystemsBiology}, or in $n$-gram word models in speech recognition~\cite{Manning_NLP}, dealing with the state space explosion is a major challenge. Also in control theory, particularly for nearly completely decomposable Markov chains, state space reduction is an important topic~\cite{Meyn_MarkovAggregation,Aldhaheri_NCDMC}.

A direct way of reducing the state space of a Markov chain is aggregation: With the help of a partition function, groups of nodes in the original transition graph are aggregated, resulting in a graph with a smaller number of nodes. The aggregated process, or \emph{aggregation}, can be any Markov chain over this smaller transition graph, depending on how the transition probabilities are chosen. Another way of reducing the state space of a Markov chain is to project realizations of the original chain through the partition function. The process thus obtained is called the projected process, or \emph{projection}. Ideally, for a given partition function, aggregated and projected process should coincide. However, as the projected process is generally not Markov, the aggregation ``closest'' to the projection is sought instead, where closeness has to be defined appropriately. In this paper, we quantify the distance between the projection and the aggregation by the Kullback-Leibler divergence rate (KLDR).

We focus on finding the optimal aggregation for a given alphabet size, i.e., on finding the partition function for which the KLDR between the projection and aggregation is minimized. 
Although, for a given partition function, the aggregation closest to the projection is easy to obtain (cf.~\cite{Meyn_MarkovAggregation} or Lemma~\ref{lem:optimalY} in this work), finding the optimal partition function remains computationally expensive because: 
(i) it requires an exhaustive search among all partitions of a given alphabet size, and 
(ii) it requires evaluating the KLDR for each candidate partition. 
In our approach, we relax the latter problem to evaluating an upper bound on the KLDR. More precisely, the aggregated Markov chain is \emph{lifted} to the original alphabet; the KLDR between the lifted and the original Markov chain provides an upper bound which can be evaluated analytically~\cite{Rached_KLDR,Meyn_MarkovAggregation}. Further relaxing the problem allows its expression in terms of the information bottleneck method~\cite{Tishby_InformationBottleneck}, casting the problem of state space reduction in terms of a widely used machine learning technique. As a result, we propose using the information bottleneck method for finding a sub-optimal partition function in a greedy manner, thus obviating the complexity of an exhaustive search for the cost of optimality.

\subsection{Contributions and Related Work} 
In control theory, model reduction and, in particular, state space aggregation of Markov models is an important topic. For example, White et al. analyzed aggregation of Markov and hidden Markov models in~\cite{White_HMMLumpable}. In particular, they presented a linear algebraic condition for lumpable chains (see Definition~\ref{def:lump} on page~\pageref{def:lump}) and determined, for a given partition function, the best aggregation in terms of the Frobenius norm. Given the transition matrix of a Markov chain, they obtained a bi-partition of its state space via alternating projection. Aldhaheri and Khalil considered optimal control of nearly completely decomposable Markov chains and adapted Howard's algorithm to work on an aggregated model~\cite{Aldhaheri_NCDMC}. The work of Jia considers state aggregation of Markov decision processes optimal w.r.t. the value function and provides algorithms which perform this aggregation~\cite{Jia_MDP}. Aggregation of Markov chains with information-theoretic cost functions 
was considered by 
Deng et al.~\cite{Meyn_MarkovAggregation} and Vidyasagar~\cite{Vidyasagar_MarkovAgg}, the first reference being the main inspiration of our work. Deng and Huang used the KLDR as a cost function to obtain a low-rank approximation of the original transition matrix via nuclear-norm regularization, thus preserving the cardinality of the state space~\cite{Deng_LowRank}. 

The idea of lifting the aggregated chain to the original state space, was used in, e.g., Deng et al.~\cite{Meyn_MarkovAggregation} and Katsoulakis et al.~\cite{Katsoulakis_CoarseGraining}. In~\cite{Katsoulakis_CoarseGraining}, the authors realized that the Kullback-Leibler divergence between the resulting Markov chains provides an upper bound on the reduction cost; however, their work is focused on continuous-time Markov chains, which makes a detailed comparison with our work difficult. Compared to~\cite{Meyn_MarkovAggregation}, our approach differs in the definition of the lifting and its consequences. More precisely, the lifting we use incorporates the one-step transition probabilities of the original chain, while the authors of~\cite{Meyn_MarkovAggregation} define lifting based only on the stationary distribution of the original chain.
Consequently, while Deng et al. maximize the redundancy of the aggregated Markov chain, the lifting proposed here minimizes \emph{information loss} in a well-defined sense.
Moreover, the upper bound we obtain is better than the upper bound obtained in~\cite{Meyn_MarkovAggregation}, and it is tight if the original chain is lumpable.

The connection to spectral graph theory observed in~\cite{Meyn_MarkovAggregation} does not apply in our case, to the best of our knowledge. More precisely, for nearly completely decomposable Markov chains, the optimal bi-partition of the alphabet is determined by the sign structure of the Fiedler vector, the eigenvector associated with the second largest eigenvalue. Despite spectral graph theory being employed for model reduction and Markov chain aggregation for some time (e.g.,~\cite{Meila_Segmentation,Runolfsson_ModelReduction}), the authors of~\cite{Meyn_MarkovAggregation} first showed a connection between this eigenvector-based aggregation method and an information-theoretic cost function.

In summary, by introducing a different lifting, we lose the connection to eigenvector-based aggregation, but instead gain the following:
\begin{enumerate}
 \item Our lifting minimizes an upper bound on the KLDR between the projection and the aggregation, subject to the requirement that the lifted chain is lumpable.
 \item The upper bound we obtain is tight if the original chain is lumpable.
 \item Minimizing the upper bound proposed by our lifting minimizes information loss in a well-defined sense; this minimization, loosely speaking, yields the partition w.r.t. which the original chain is ``most lumpable''.
 \item Relaxating the cost function allows applying the information bottleneck method for state space aggregation.
\end{enumerate}
The connection to the information bottleneck method is most interesting: Recently, Vidyasagar investigated a metric between distributions on sets of different cardinalities~\cite{Vidyasagar_Distribution}, a problem very similar to the one considered in this work. He proposed an information-theoretic metric called the variation of information, and showed that the optimal reduced-order distribution on a set of given cardinality is obtained by \emph{projecting} the original distribution. Specifically, the reduced-order distribution should have maximal entropy, which is equivalent to requiring that the partition function induces the minimum information loss; a sub-optimal solution to this problem is given by the information bottleneck method, cf.~\cite{Geiger_Relevant_arXiv}.

We furthermore provide new insight in some aspects of the lifting proposed in~\cite{Meyn_MarkovAggregation}:
\begin{enumerate}
 \item The KLDR between the original Markov chain and its aggregation, as defined in~\cite{Meyn_MarkovAggregation}, is an upper bound on the KLDR between the projection and the aggregation.
 \item Following~\cite{Kemeny_FMC}, we introduce a compact matrix notation for the lifting introduced in~\cite{Meyn_MarkovAggregation}, allowing us to provide new proofs for some of the results shown there.
\end{enumerate}
                                                                                                                                                        
In works related to (graph) clustering, information-theoretic cost functions are often used for error quantification.
In particular, in~\cite{Raj_Clustering}, the authors use the information bottleneck method for partitioning a graph via assuming continuous-time graph diffusion. 
Moreover, in~\cite{Tishby_MarkovRelax} and~\cite{Friedman_Clustering} pairwise distance measures between data points were used to define a stationary Markov chain, whose statistics are then used for clustering the data points. While~\cite{Tishby_MarkovRelax} applies the information bottleneck method and obtains a result very similar to ours, the authors do not describe its importance for Markov chain aggregation. 
In~\cite{Friedman_Clustering}, the authors employ the same cost function as~\cite{Meyn_MarkovAggregation} and present an iterative algorithm similar to the agglomerative information bottleneck method~\cite{Slonim_AgglomIB}.
While their work focuses on pairwise clustering, they conclude by stating that their results can be employed for Markov chain aggregation. Most recently, the authors of~\cite{Xu_DA} proposed graph clustering by defining a dissimilarity function between the original and the aggregated graph and subsequently applying deterministic annealing to find the best clustering. They define a composite graph to cope with the problem of comparing graphs with different sizes and apply their results to Markov chain aggregation using KLDR as a dissimilarity measure.

Although the present work focuses on stationary Markov chains, we conjecture that our results can be generalized to time-homogeneous Markov chains with a starting distribution different from the invariant distribution, since they are still \emph{stationary in the asymptotic mean}, or AMS~\cite{Kieffer_MarkovChannelsAMS}. A more detailed discussion of AMS can be found in~\cite{Gray_Probability}.

The extension to \emph{stochastic aggregations}, i.e., an aggregation where the state space reduction is not performed by a deterministic partition function, but rather by a stochastic mapping, is not immediately possible, at least to the best of our knowledge. While a result similar to our Lemma~\ref{lem:kldbound} should hold also for stochastic mappings (cf.~\cite[Ch.~4.4]{Cover_Information2}), it is not clear how the stochastic aggregation should be lifted to a Markov chain on the original state space. Since deterministic mappings are preferable for their simplicity, we leave the topic of stochastic aggregations for future investigation.

\subsection{Outline of the Paper}
We start by introducing notation and information-theoretic quantities (Sections~\ref{ssec:RVs} and~\ref{ssec:ITQuant}) and their application to Markov chains (Section~\ref{ssec:Markov}) and functions of Markov chains (Section~\ref{ssec:MarkovFunction}; introducing also the notion of lumpability). Turning to the problem of state space aggregation in Section~\ref{sec:originalApproach}, we restate results linked to the lifting method proposed by~\cite{Meyn_MarkovAggregation} (Section~\ref{sec:pivecLifting}) and present an alternative and its properties (Section~\ref{sec:Plifting}). Section~\ref{sec:infoloss} connects the proposed lifting method to the notion of relevant information loss recently introduced in~\cite{Geiger_Relevant_arXiv}; we exploit this connection in Section~\ref{ssec:ibmethod} to show how the information bottleneck method can be employed for state space reduction. A few small examples are contained in Section~\ref{sec:examples}. The final section, 
Section~\ref{sec:bioexample}, is devoted to a biologically inspired example.

\section{Notation, Preliminaries, and Setup}\label{sec:prelim}
\subsection{Random Variables and Stochastic Processes}\label{ssec:RVs}
Let $(\Omega,\mathfrak{B},\mathrm{Pr})$ denote the probability space on which all random variables (RVs) and stochastic processes are defined. We denote RVs by upper case letters, e.g., $Z$, their (finite) alphabet by calligraphic letters, e.g., $\dom{Z}$, and realizations by lower case letters, e.g., $z$, where $z\in\dom{Z}$. For an index set $\mathbb{I}\subset\mathbb{N}$ with finite cardinality $\card{\mathbb{I}}$, let $Z_\mathbb{I}:=\{Z_i\}_{i\in\mathbb{I}}$; in particular, we abbreviate $Z_m^n:=\{Z_{m},Z_{m+1},\dots,Z_n\}$. The probability mass function (PMF) of $Z$ is denoted by $p_Z$, where
\begin{equation}
 \forall z\in\dom{Z}{:} \quad p_Z(z) := \Prob{Z=z}.
\end{equation} 
The joint PMF $p_{Z_\mathbb{I}}$ of $Z_\mathbb{I}$ and the conditional PMF $p_{Z_\mathbb{I}|Z_\mathbb{J}}$ of $Z_\mathbb{I}$ given $Z_\mathbb{J}$ are defined similarly.

In this work, discrete-time, one-sided random process are denoted by bold-faced letters, e.g., $\Zvec$, and their (random) samples are indexed by the set of natural numbers, i.e., $\{Z_1,Z_2,\dots\}$. We assume each RV $Z_i$ takes values from the same, finite, alphabet $\dom{Z}$. The random processes considered in this work are \emph{stationary}. In particular, the marginal distribution of $Z_k$ is equal for all $k$ and shall be denoted as $p_Z$.

\subsection{Information-Theoretic Quantities}\label{ssec:ITQuant}
In the remainder of this work we will need
\begin{definition}[Information-Theoretic Quantities{~\cite[Ch,~2 \& 4]{Cover_Information2}}]\label{def:ITQuant}
The \emph{(joint) entropy} of a collection of RVs $Z_\mathbb{I}$, the \emph{conditional entropy} of $Z_\mathbb{I}$ given $Z_\mathbb{J}$, and the \emph{mutual information} between $Z_\mathbb{I}$ and $Z_\mathbb{J}$ are
\begin{subequations}
\begin{align}
 \ent{Z_\mathbb{I}}&:=-\sum_{z_\mathbb{I}\in\dom{Z}^{\card{\mathbb{I}}}} p_{Z_\mathbb{I}}(z_\mathbb{I})\log p_{Z_\mathbb{I}}(z_\mathbb{I})\\
 \ent{Z_\mathbb{I}|Z_\mathbb{J}} &:=\ent{Z_{\mathbb{I}\cup\mathbb{J}}}-\ent{Z_\mathbb{J}}=\ent{Z_{\mathbb{I}},Z_\mathbb{J}}-\ent{Z_\mathbb{J}}\\
\mutinf{Z_\mathbb{I};Z_\mathbb{J}}&:=\ent{Z_\mathbb{J}}+\ent{Z_\mathbb{I}}-\ent{Z_{\mathbb{I}},Z_\mathbb{J}}.
\end{align}

The \emph{entropy rate} and the \emph{redundancy rate} of a stationary stochastic process $\Zvec$ are
\begin{align}
 \entrate{Z}&:=\limn\frac{1}{n} \ent{Z_1^{n}} = \limn \ent{Z_n|Z_{1}^{n-1}}\\
 \redrate{Z}&:=\ent{Z}-\entrate{Z}\stackrel{(a)}{\geq}0
\end{align}
\end{subequations}
where $\ent{Z}$ is the entropy of the marginal distribution of $\Zvec$ and where $(a)$ is due to the fact that conditioning reduces entropy~\cite[Thm.~2.6.5]{Cover_Information2}. 
\end{definition}

The redundancy rate is a measure of statistical dependence between the current sample and its past: For a process of independent, identically distributed RVs, $\entrate{Z}=\ent{Z}$ and $\redrate{Z}=0$. Conversely, for a completely predictable process, $\entrate{Z}=0$ and $\redrate{Z}=\ent{Z}$. In other words, the higher the redundancy rate, the lower the entropy rate and, thus, the less information is conveyed by the process in each time step.

We need another definition for the development of our results:
\begin{definition}[Kullback-Leibler Divergence Rate]\label{def:KLDR}
 The Kullback-Leibler divergence rate (KLDR) between two stationary stochastic processes $\Zvec$ and $\Zvec'$ on the same finite alphabet $\dom{Z}$ is~\cite[Ch.~10]{Gray_Entropy}
\ifthenelse{\single=1}{
\begin{equation}
 \kldrate{Z}{Z'}:=\limn \frac{1}{n} \kld{Z_1^n}{{Z'}_1^{n}}
= \limn \frac{1}{n}\sum_{z_1^n\in\dom{Z}^n} p_{Z_1^n}(z_1^n) \log \frac{p_{Z_1^n}(z_1^n)}{p_{{Z'}_1^n}(z_1^n)}
\end{equation}
}
{\begin{align}
 \kldrate{Z}{Z'}&:=\limn \frac{1}{n} \kld{Z_1^n}{{Z'}_1^{n}}\notag\\
&= \limn \frac{1}{n}\sum_{z_1^n\in\dom{Z}^n} p_{Z_1^n}(z_1^n) \log \frac{p_{Z_1^n}(z_1^n)}{p_{{Z'}_1^n}(z_1^n)}
\end{align}}
whenever the limit exists and if $p_{Z_1^n}\ll  p_{{Z'}_1^n}$ for all $n$, i.e., if for all $n$ and all $z_1^n$,
\begin{equation}
   p_{{Z'}_1^n}(z_1^n) = 0 \Rightarrow p_{Z_1^n}(z_1^n) = 0.
\end{equation}
\end{definition}

The limit exists, e.g., between a stationary stochastic process and a time-homogeneous Markov chain~\cite{Gray_Entropy} as well as between Markov chains (not necessarily stationary or irreducible)~\cite{Rached_KLDR}. Roughly speaking, the KLDR between a process $\Zvec$ and its model $\Zvec'$ quantifies the number of bits necessary per time step to correct the model distribution to arrive at the true process distribution.

\subsection{Markov Chains}\label{ssec:Markov}
Let $\Xvec$ be a regular, i.e., irreducible and aperiodic, time-homogeneous Markov chain on the finite alphabet $\dom{X}=\{1,\dots,N\}$ (see~\cite{Kemeny_FMC} for terminology). Its behavior is uniquely determined by its transition matrix $\Pvec=\{P_{ij}\}$, where $ P_{ij} := \Prob{X_n=j|X_{n-1}=i}$. The unique invariant distribution vector $\muvec$ with its $i$-th component given by
\begin{equation}
 \mu_i:=p_X(i)=\Prob{X_k=i}>0
\end{equation} 
satisfies $\muvec^T=\muvec^T\Pvec$~\cite[Thm.~4.1.6]{Kemeny_FMC}. For such a Markov chain we use the shorthand notation $\Xvec\sim\Mar{\dom{X},\Pvec,\muvec}$. We assume furthermore that $\Xvec$ is stationary, i.e., its initial distribution coincides with the invariant distribution. 

With Definition~\ref{def:ITQuant}, the entropy and the entropy rate of $\Xvec$, as well as the KLDR between two Markov chains $\Xvec$ and $\Xvec'$ on the same alphabet $\dom{X}$ with transition matrices $\Pvec$ and $\Pvec'$ are~\cite[p.~77]{Cover_Information2},~\cite{Rached_KLDR}
\begin{align}
  \ent{X} &= -\sum_{i\in\dom{X}} \mu_i\log\mu_i\\
  \entrate{X} &= \ent{X_1|X_0} = - \sum_{i,j\in\dom{X}} \mu_i P_{ij}\log P_{ij}\label{eq:entX}\\
\kldrate{X}{X'}&= \sum_{i,j\in\dom{X}} \mu_i P_{ij}\log\frac{P_{ij}}{P'_{ij}} \label{eq:kldMarkov}
\end{align}
respecively, provided that $P'_{ij}=0$ implies $P_{ij}=0$ ($\Pvec\ll\Pvec'$).

\subsection{Functions of Markov Chains}\label{ssec:MarkovFunction}
We partition the alphabet $\dom{X}$ of the Markov chain $\Xvec$ by a surjective function $g{:}\ \dom{X}\to\dom{Y}$, where $\dom{Y}=\{1,\dots,M\}$. In other words, $g$ induces a partition of $\dom{X}$ by the preimages\footnote{Given a state $j\in \dom{Y}$, with slight abuse of notation we write $\preim{j}$ for the preimage of $j$ under $g$, that is, $\preim{j}:=\preim{\{j\}}=\{i\in\dom{X}\mid g(i)=j\}$.} of the elements of $\dom{Y}$. Projecting $\Xvec$ through the function, i.e., $Y_n:=g(X_n)$, defines another stochastic process $\Yvec$; in what follows we call this process the \emph{projected process}, or simply the \emph{projection} of $\Xvec$ (see Fig.~\ref{fig:connection}).

If $\Xvec$ is stationary, then so is $\Yvec$. The following inequalities 
\begin{IEEEeqnarray}{RCL}
 \ent{Y}&\leq&\ent{X}\\\entrate{Y}&\leq&\entrate{X}\\\redrate{Y}&\leq&\redrate{X}
\end{IEEEeqnarray}
hold by the data processing inequality~\cite{Cover_Information2,Pinsker_InfoEngl} and by~\cite{Watanabe_InfoLoss}.

It is well known that $\Yvec$ is not necessarily Markov. The case where $\Yvec$ is a regular, time-homogeneous Markov chain, gives rise to the notion of lumpability: 
\begin{definition}[Lumpability~\cite{Kemeny_FMC}]\label{def:lump}
 A Markov chain $\Xvec\sim\Mar{\dom{X},\Pvec,\muvec}$ is \emph{lumpable} w.r.t. a function $g$, iff the process $\Yvec$ is a regular, time-homogeneous Markov chain with alphabet $\dom{Y}$, transition matrix $\Qvec$, and invariant distribution $\nuvec$, i.e., iff $\Yvec\sim\Mar{\dom{Y},\Qvec,\nuvec}$, for every initial distribution of $\Xvec$.
\end{definition}

In order to present conditions under which a Markov chain is lumpable, we need the following matrices: Let $\Vvec$ be an $N\times M$ matrix with $V_{ij}:=1$ if $i\in\preim{j}$ and zero otherwise (thus, every row contains exactly one 1). Furthermore, $\Uvec^{\pivec}$ is an $M\times N$ matrix with zeros in the same positions as $\Vvec^T$, but with otherwise positive row entries which sum to one. In other words, with $\pivec$ being a positive probability vector,
\begin{equation}\label{eq:umatdef}
 U^{\pivec}_{ij}:=
\begin{cases}
 \frac{\pi_j}{\sum_{k\in\preim{i}}\pi_k}, &\text{ if } j\in\preim{i}\\
 0, & \text{else}
\end{cases}.
\end{equation}

\newcommand{\zetavec}{\boldsymbol{\zeta}}
\begin{lem}[Conditions for Lumpability]\label{lem:lump} \rm
 A stationary Markov chain $\Xvec\sim\Mar{\dom{X},\Pvec,\muvec}$ is lumpable w.r.t. $g$ iff for every positive probability vector $\zetavec$
\begin{equation}\label{eq:condstrong}
  \Vvec\Uvec^{\zetavec}\Pvec\Vvec = \Pvec\Vvec. 
\end{equation}
Then, $\Yvec\sim\Mar{\dom{Y},\Qvec,\nuvec}$ with $\nuvec^T=\muvec^T\Vvec$ and
\begin{equation}
 \Qvec=\Uvec^{\muvec}\Pvec\Vvec.
\end{equation}
\end{lem}

\begin{proof}
See~\cite[Thm.~6.3.5 \& Example~6.3.3]{Kemeny_FMC} and note that
\begin{equation}
 \nuvec^T=\nuvec^T\Qvec = \nuvec^T\Uvec^{\muvec}\Pvec\Vvec = \muvec^T\Pvec\Vvec=\muvec^T\Vvec.
\end{equation}
\end{proof}

The corresponding result for continuous-time Markov chains on a countable alphabet has been proven in~\cite[Thm.~2]{Petrov_Lumpability}. 

While the KLDR between two Markov chains is easy to compute, for the KLDR between a (non-Markov) function of a Markov chain and Markov chain no closed-form solution is available. In special cases, however, the former can act as an upper bound on the latter, provided the Markov chains are chosen appropriately. We make this precise in

\begin{lem}\label{lem:kldbound} \rm
Let $\Xvec$ and $\Xvec'$ be stationary, time-homogeneous, regular Markov chains on the same alphabet $\dom{X}$ with transition matrices $\Pvec$ and $\Pvec'$. Let $\Pvec'\gg \Pvec$. We define two processes $\Yvec$ and $\Yvec'$ by $Y_n:=g(X_n)$ and $Y'_n:=g(X'_n)$, $g{:}\ \dom{X}\to\dom{Y}$. Let additionally $\Xvec'$ be lumpable w.r.t. $g$. We have
 \begin{equation}
  \kldrate{X}{X'} \geq \kldrate{Y}{Y'}.
 \end{equation}
\end{lem}

\begin{proof}
The inequality follows from the fact that the Kullback-Leibler divergence reduces under measurements (e.g.,~\cite[Cor.~3.3]{Gray_Entropy} or~\cite[Ch.~2.4]{Pinsker_InfoEngl}), i.e., that for all $n$,
\begin{equation}
\kld{X_1^n}{{X'}_1^{n}} \geq \kld{Y_1^n}{{Y'}_1^{n}}.
\end{equation}
It thus remains to show that the limits exist.

Since $\Pvec'\gg \Pvec$, $\kldrate{X}{X'}$ exists and equals~\eqref{eq:kldMarkov}, cf.~\cite{Rached_KLDR}. Since $\Xvec'$ is lumpable, $\Yvec'$ is a regular, time-homogeneous Markov chain. Moreover, from $\Pvec'\gg \Pvec$ it follows that the process distribution of $\Yvec$ is absolutely continuous w.r.t. the process distribution of $\Yvec'$. This ensures the existence of $\kldrate{Y}{Y'}$~\cite[Lem.~10.1]{Gray_Entropy} and completes the proof.
\end{proof}

\section{Problem Statement} 
\label{sec:originalApproach}
\begin{figure}[t] 
\centering
  \begin{pspicture}[showgrid=false](-1,0.5)(6,6)
    \psset{style=Arrow}
    \pssignal(0,5){x}{$\Xvec\sim\Mar{\dom{X},\Pvec,\muvec}$}
    \pssignal(0,1){y}{$\Yvecg$}
    \pssignal(5,1){yp}{$\Yvecg'\sim\Mar{\dom{Y},\Qvec,\nuvec}$}
    \pssignal(5,5){xp}{$\Xvec'\sim\Mar{\dom{X},\Pvec',\muvec'}$}
    \ncline{<-}{y}{x} \aput{:U}{Projection $g$}
    \ncline[style=Dash]{x}{yp}\aput{:U}{Aggregation}
    \pnode(4.8,1.4){yp1}\pnode(4.8,4.6){xp1}\pnode(5.2,1.4){yp2}\pnode(5.2,4.6){xp2}
    \ncline{yp1}{xp1}\bput[-15pt]{:U}{lifting}
    \ncline{<-}{yp2}{xp2}\bput[5pt]{:U}{Projection $g$}
    \ncline[style=Dash]{<->}{y}{yp}\Aput{$\bar{D}(\Yvecg || \Yvecg')$}
    \ncline[style=Dash]{<->}{x}{xp}\Aput{$\bar{D}(\Xvec || \Xvec')$}
\end{pspicture}
\caption{Illustration of the problem: Assume a Markov chain $\Xvec$ is given. We are interested in finding an aggregation of $\Xvec$, i.e., a Markov chain $\Yvecg'$ on a partition of the alphabet of $\Xvec$. This partition defines a function $g$ (and vice-versa), which allows us to define a process $\Yvecg$ (via $Y_{g,n}:=g(X_n)$), the projection of $\Xvec$. Note that $\Yvecg$ need not be Markov. Lifting $\Yvecg'$ yields a Markov chain on the original alphabet, which can be projected to $\Yvecg'$ using the function $g$.}
\label{fig:connection}
\end{figure}
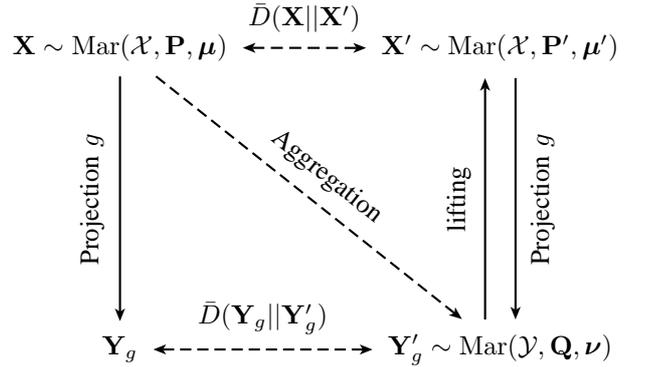

Throughout the remainder of this work we will stick to the following
\begin{assumption} \label{assumption:in}
The discrete-time Markov chain $\Xvec\sim\Mar{\dom{X},\Pvec,\muvec}$ is stationary, i.e., the initial distribution equals its invariant distribution $\muvec$. The alphabet of $\Xvec$ is $\dom{X} = \{1,\ldots,N\}$ and the partition function is $g{:}\ \dom{X}\to\dom{Y}=\{1,\dots,M\}$ with $1<M<N$. The \emph{$g$-projection} of $\Xvec$ is the stationary process $\Yvecg$ over the alphabet $\dom{Y}$, whose samples are defined by
\begin{equation}
Y_{g,n} := g(X_n).
\end{equation}
\end{assumption}

We are interested in performing model reduction by employing information-theoretic cost functions.  
In particular, we specify the $M$-partition problem, equivalently defined in~\cite{Meyn_MarkovAggregation}: 

\newcommand{\argmin}{\operatornamewithlimits{argmin}}
\begin{definition}[$M$-partition problem]
Given $\Xvec$ and $g$ as in Assumption \ref{assumption:in}, the \emph{$M$-partition problem} searches for the partition function $g$ such that the KLDR between the $g$-projection of $\Xvec$ and its best Markov approximation is minimal, i.e., it solves
\begin{equation} \label{eq:problem2}
\begin{split}
\argmin_{g\in[\dom{X}\to\dom{Y}]} \min_{\Yvec'} \{ \bar{D}(\Yvecg || \Yvec') \mid \Yvec'\hbox{ is Markov}\}.
\end{split}
\end{equation}
\end{definition}

For a fixed partition function $g$ the best Markov approximation (in the sense of the KLDR) of the $g$-projection $\Yvecg$ can be found analytically. With the matrix notation introduced in Section~\ref{ssec:MarkovFunction} we present
\begin{lem}\label{lem:optimalY} \rm
Given $\Xvec$, $g$ and $\Yvecg$ as in Assumption \ref{assumption:in}, let $\Yvecg'$ denote the best Markov approximation of the $g$-projection in the sense of the KLDR, i.e., 
\begin{equation}
\Yvecg' := \argmin_{\Yvec'} \{ \bar{D}(\Yvecg || \Yvec') \mid \Yvec'\hbox{ is Markov}\}.
\end{equation}
Then, $\Yvecg'\sim\Mar{\dom{Y},\Qvec,\nuvec}$ with $\nuvec^T=\muvec^T\Vvec$ and
\begin{equation}
 \Qvec = \Uvec^{\muvec}\Pvec\Vvec,
\end{equation}
which is a matrix notation for 
\begin{equation}
 Q_{kl} = \frac{\sum_{i\in\preim{k}}\sum_{j\in\preim{l}} \mu_iP_{ij}}{\sum_{i\in\preim{k}}\mu_i}, \;\; k,l\in\dom{Y}.
\label{eq:Qmatlong}
\end{equation}
\end{lem}

\begin{proof}
See~\cite[Cor.~10.4]{Gray_Entropy}.
\end{proof}

From now on, we keep the notation $\Yvecg'$ for the optimal aggregation (see Fig.~\ref{fig:connection}) of $\Xvec$, given a partition function $g$. 

\begin{remark}\rm
 The same aggregation was declared being optimal in~\cite{Weinan_Aggregation}, although by using a different cost function. Also~\cite[Thm.~1]{Meyn_MarkovAggregation} declares this aggregation as being optimal, although the cost function there is the KLDR between the original chain $\Xvec$ and the \emph{lifting} of $\Yvecg'$ (see Section~\ref{sec:pivecLifting} below).
\end{remark}

One thus obtains the transition matrix $\Qvec$ of the optimal Markov model $\Yvecg'$ from the joint distribution of two consecutive samples of $\Yvecg$. If this joint distribution completely specifies the process $\Yvecg$, then $\Yvecg\equiv\Yvecg'$, i.e., $\Xvec$ is lumpable (cf.~Lemma~\ref{lem:lump}). Note further that since $\Pvec$ is the transition matrix of a regular Markov chain, so is $\Qvec$~\cite[p.~140]{Kemeny_FMC}. 

We can now define the aggregation error of $\Xvec$ w.r.t. $g$:
\begin{definition}[Aggregation error]\label{def:aggerror}
Given $\Xvec$, $g$ and $\Yvecg$ as in Assumption \ref{assumption:in}, and $\Yvecg'$ as in Lemma \ref{lem:optimalY}. Then, 
\begin{equation}
\kldr{\Yvecg}{\Yvecg'}
\end{equation} 
is the \emph{aggregation error} of $\Xvec$ w.r.t. $g$.
\end{definition}
It immediately follows that the aggregation error is zero if $\Xvec$ is lumpable.


Following~\cite{Meyn_MarkovAggregation}, we split the $M$-partition problem into two sub-problems: finding the best Markov approximation of the projected process $\Yvecg$ (to which Lemma~\ref{lem:optimalY} provides the solution), and minimizing the aggregation error over all partition functions $g$ with a range of cardinality $M$. Thus, the optimization problem stated in (\ref{eq:problem2}) translates to finding
\begin{equation}
\argmin_{g\in[\dom{X}\rightarrow \dom{Y}]} \kldr{\Yvecg}{\Yvecg'}.
\end{equation}

\section{$\pivec$-lifting: Bounding the Aggregation Error}
\label{sec:pivecLifting}
Often, a direct evaluation of the aggregation error in Definition~\ref{def:aggerror} is mathematically cumbersome, since $\Yvecg$ is not necessarily Markov. The authors of~\cite{Meyn_MarkovAggregation} therefore suggested to \emph{lift} the aggregation $\Yvecg'$ to a Markov chain $\Xvec'$ over the alphabet $\dom{X}$, which subsequently allows a computation of the KLDR. 

The questions is now, whether there is a relation between the KLDR between $\Xvec$ and the lifted chain $\Xvec'$, and the aggregation error, $\kldr{\Yvecg}{\Yvecg'}$. Relying on Lemma~\ref{lem:kldbound}, we will answer this question affirmatively.

\begin{definition}[$\pivec$-lifting~{\cite[Def.~2]{Meyn_MarkovAggregation}}]\label{def:pilifting}                                                                                                                                                       
Given $\Xvec$, $g$ and $\Yvecg$ as in Assumption \ref{assumption:in}, $\Yvecg'\sim\Mar{\dom{Y},\Qvec,\nuvec}$ as in Lemma~\ref{lem:optimalY}, and $\pivec$ a positive probability distribution over the alphabet $\dom{X}$. The \emph{$\pivec$-lifting} of $\Yvecg'$ w.r.t. $g$, denoted by $\Xvec'^{\pivec}_g$, is a Markov chain over the alphabet $\dom{X}$ with transition matrix
\begin{equation}
 \Pvec' := \Vvec\Qvec\Uvec^{\pivec},
\end{equation}
which is a matrix notation for 
\begin{equation}
 {P'}_{ij} = \frac{\pi_j}{\sum_{k\in \preim{g(j)}}\pi_k} Q_{g(i)g(j)},\;\; i,j \in\dom{X}.
\end{equation}
\end{definition}

\begin{remark} \rm
An equivalent lifting method is suggested in~\cite{Weinan_Aggregation} and~\cite{Vidyasagar_KLDRate,Vidyasagar_MarkovAgg}.
\end{remark}

We conclude this section by presenting the elementary properties of $\pivec$-lifting. Properties 1), 2), and 3) appear also in~\cite{Meyn_MarkovAggregation}; the proofs can be found in Appendix~\ref{app:pivecProperties} and, in contrast to the proofs in~\cite{Meyn_MarkovAggregation}, appear in short matrix notation. To the best of the authors' knowledge, properties 4) and 5) are proved for the first time here.

\begin{prop}[Properties of $\pivec$-lifting] \label{prop:liftingProperties}
Given $\Xvec$, $g$ and $\Yvecg$ as in Assumption \ref{assumption:in},  
$\Yvecg'$ as in Lemma \ref{lem:optimalY}, and
$\pivec$ some distribution over $\dom{X}$. Then, the $\pivec$-lifting $\Xvec'^{\pivec}_g$ satisfies
\begin{enumerate}
\item $\Xvec'^{\pivec}_g$ is lumpable w.r.t. $g$ (and $\Yvecg'$ is the resulting $g$-projection);
\item The invariant distribution of $\Xvec'^{\muvec}_g$ is $\muvec$;
\item $\muvec = \argmin_{\pivec}\kldr{\Xvec}{\Xvecg'^{\pivec}}$;
\item $\Pvec'\gg\Pvec$;
\item $\kldr{\Yvecg}{\Yvecg'}\leq\kldr{\Xvec}{\Xvecg'^{\muvec}}$.
\end{enumerate}
\end{prop}

\section{A Better Bound via $\Pvec$-lifting}
\label{sec:Plifting}
In Section~\ref{sec:pivecLifting}, we showed that the KLDR between $\Xvec$ and the $\pivec$-lifting with $\pivec=\muvec$, $\Xvecg'^{\muvec}$, provides an upper bound on the aggregation error for a given partition function $g$. Unfortunately, the bound is loose in the sense that for $\kldr{\Yvecg}{\Yvecg'} = 0$, we may have $\kldr{\Xvec}{\Xvecg'^{\muvec}} > 0$; see also~\cite{Vidyasagar_MarkovAgg}. One of the reasons for this disadvantage of $\pivec$-lifting is that, by construction, the lifted process $\Xvecg'^{\pi}$ does not contain information about the transition probabilities between states of $\Xvec$. We therefore propose a lifting which takes into account the transition matrix $\Pvec$ of the original process.

\begin{definition}[$\Pvec$-lifting]\label{def:plift}
Given $\Xvec$, $g$ and $\Yvecg$ as in Assumption~\ref{assumption:in} and $\Yvecg'\sim\Mar{\dom{Y},\Qvec,\nuvec}$ as in Lemma~\ref{lem:optimalY}. The \emph{$\Pvec$-lifting} of $\Yvecg'$ w.r.t. $g$, denoted by $\Xvecg'^{\Pvec}$, is a Markov chain over the alphabet $\dom{X}$ with a transition matrix $\hat{\Pvec}$ given by
\begin{equation}
 \hat{P}_{ij} := 
\begin{cases}
	\frac{P_{ij}}{\sum_{k\in g^{-1}(g(j))}P_{ik}} Q_{g(i)g(j)}, &\text{ if } \displaystyle\sum_{k\in g^{-1}(g(j))}P_{ik}>0\\
	\frac{1}{\card{g^{-1}(g(j))}}Q_{g(i)g(j)}, &\text{ if } \displaystyle\sum_{k\in g^{-1}(g(j))}P_{ik}=0
\end{cases}.\label{eq:myLift}
\end{equation}
\end{definition}

One of the main contributions of this paper is to show that the KLDR between $\Xvec$ and the $\Pvec$-lifting $\Xvecg'^{\Pvec}$ yields a better bound than the one obtained using $\pivec$-lifting. In Appendix~\ref{app:PProperties} we prove

\begin{thm}[Properties of $\Pvec$-lifting]\label{thm:PliftingProperties} 
Given $\Xvec$, $g$ and $\Yvecg$ as in Assumption~\ref{assumption:in} and $\Yvecg'$ as in Lemma~\ref{lem:optimalY}. Then, the $\Pvec$-lifting $\Xvecg'^{\Pvec}$ satisfies
\begin{enumerate}
\item $\Xvecg'^{\Pvec}$ is lumpable w.r.t. $g$ (and $\Yvecg'$ is the resulting $g$-projection);
\item $\hat{\Pvec}\gg\Pvec$;
\item (minimizer)
    \begin{equation}
    \Xvecg'^{\Pvec} = \argmin_{\hat{\Xvec}: \Yvecg' \text{ is $g$-projection of } \hat{\Xvec}} \kldr{\Xvec}{\hat{\Xvec}}
    \end{equation}
\item (better bounds than $\pi$-lifting)
    \begin{equation}
      \kldr{\Yvecg}{\Yvecg'}\leq \kldr{\Xvec}{\Xvecg'^{\Pvec}}\leq\kldr{\Xvec}{\Xvecg'^{\muvec}}
    \end{equation}
\item (tight bounds) If $\Xvec$ is lumpable w.r.t. $g$,
    \begin{equation}
      \kldr{\Yvecg}{\Yvecg'} =0 \Leftrightarrow \kldr{\Xvec}{\Xvecg'^{\Pvec}}=0.
    \end{equation}
\end{enumerate}
\end{thm}

Tightness follows from the fact that for a lumpable $\Xvec$, the $\Pvec$-lifting yields $\hat{\Pvec}=\Pvec$; the invariant distribution of $\hat{\Pvec}$ trivially coincides with $\muvec$, the invariant distribution of $\Pvec$. In general, however, the invariant distribution of $\hat{\Pvec}$ differs from $\muvec$, contrasting the corresponding result for the $\pivec$-lifing (cf.~Proposition~\ref{prop:liftingProperties}, property 2).

Interestingly, the restriction to lumpable chains for the tightness result cannot be dropped: There are Markov chains $\Xvec$ which are lumpable in a weaker sense (i.e., not for all initial distributions but, e.g., only for the invariant distribution) for which consequently the aggregation error vanishes, but for which the $\Pvec$-lifting does not yield $\hat{\Pvec}=\Pvec$. A simple example of such a chain is given in~\cite[p.~139]{Kemeny_FMC} (cf.~Section~\ref{ex:weaklyOnly}).

As this theorem shows, $\Pvec$-lifting yields the best upper bound on the aggregation error achievable for Markov chains over the alphabet $\dom{X}$. This can also be explained intuitively, by expanding the KLDR as
\ifthenelse{\single=1}{
\begin{IEEEeqnarray}{RCL}
 \kldr{\Xvec}{\Xvecg'^{\Pvec}}&=& \sum_{i,j\in\dom{X}} \mu_i P_{ij}\log\frac{P_{ij}}{\hat{P}_{ij}}\\
&=& \sum_{i,j\in\dom{X}} \mu_i P_{ij}\log\frac{{\sum_{k\in\dom{S}_j}P_{ik}}}{ Q_{g(i)g(j)}}\\
&\stackrel{(a)}{=}& \sum_{i,j\in\dom{X}} \mu_i P_{ij} \log\frac{\left(\sum_{k\in\dom{S}_i}\mu_k\right)\left(\sum_{l\in\dom{S}_j}P_{il}\right)}{\sum_{k\in\dom{S}_i}\mu_k\sum_{l\in\dom{S}_j}P_{kl}}\\
&=&\ent{Y_{g,n}|Y_{g,n-1}}-\ent{Y_{g,n}|X_{n-1}}\label{eq:diffbounds}
\end{IEEEeqnarray}}
{
\begin{IEEEeqnarray}{RCL}
 \kldr{\Xvec}{\Xvecg'^{\Pvec}}&=& \sum_{i,j\in\dom{X}} \mu_i P_{ij}\log\frac{P_{ij}}{\hat{P}_{ij}}\\
&=& \sum_{i,j\in\dom{X}} \mu_i P_{ij}\log\frac{{\sum_{k\in\dom{S}_j}P_{ik}}}{ Q_{g(i)g(j)}}\\
&\stackrel{(a)}{=}& \sum_{i,j\in\dom{X}} \mu_i P_{ij} \log\frac{\left(\sum_{k\in\dom{S}_i}\mu_k\right)\left(\sum_{l\in\dom{S}_j}P_{il}\right)}{\sum_{k\in\dom{S}_i}\mu_k\sum_{l\in\dom{S}_j}P_{kl}}\notag\\\\
&=&\ent{Y_{g,n}|Y_{g,n-1}}-\ent{Y_{g,n}|X_{n-1}}\label{eq:diffbounds}
\end{IEEEeqnarray}
}
where $(a)$ is due to Lemma~\ref{lem:optimalY}. The last line corresponds to the difference between the upper and lower bounds on the entropy rate of a function of a Markov chain~\cite[Thm.~4.5.1]{Cover_Information2}; equality of these bounds implies Markovity of $\Yvecg$, i.e., lumpability of $\Xvec$ w.r.t. $g$~\cite[Thm.~9]{GeigerTemmel_kLump}. In other words, minimizing this cost function yields the function $g$ for which the projected process $\Yvecg$ is ``as Markov as possible''.

\newcommand{\redr}[1]{\bar{R}(#1)}
\newcommand{\entr}[1]{\bar{H}(#1)}

\section{$\Pvec$-Lifting and Information Loss}\label{sec:infoloss}\label{ssec:plift}
We now analyze how the KLDR between the original process and a $\Pvec$-lifted process connects with the information loss induced by the projection function $g$. This parallels the analysis in~\cite{Meyn_MarkovAggregation}, claiming that $\pivec$-lifting maximizes the redundancy rate\footnote{The reader should not be misled by the fact that the redundancy rate satisfies $\redr{\Yvecg'}=\mutinf{Y_{g,0};Y_{g,1}}$, i.e., that it is formulated as a mutual information. The fact that the current sample $Y_{g,0}$ shares much information with the future sample $Y_{g,1}$ only emphasizes that the process is redundant, i.e., that it conveys \emph{little new information in each time step.}} of the aggregated process. Interestingly, the cost function induced by $\Pvec$-lifting \emph{minimizes} a special notion of information loss introduced recently. Moreover, the latter analysis paves the way for solving the state space reduction problem using information-theoretic algorithms, such as the information bottleneck 
method (cf.~Section~\ref{ssec:ibmethod}).

\begin{definition}[Relevant Information Loss~\cite{Geiger_Relevant_arXiv}]\label{def:relevantLoss}
 Let $X$ be an RV with finite alphabet $\dom{X}$, and let $Y:=g(X)$. Let $S$ be another RV with alphabet $\dom{S}$ representing \emph{relevant information}. The information loss \emph{relevant w.r.t. $S$} is
\begin{equation}
 \loss[S]{X\to Y} = \mutinf{S;X}-\mutinf{S;Y}=\mutinf{X;S|Y}.\label{eq:relevantloss}
\end{equation}
\end{definition}

A simple example to illustrate this notion of information loss is the following: Let $S$ be a binary signal, and let $X$ be this signal superimposed by noise, e.g., the output of a noisy communications channel. By passing $X$ through a function $g$, e.g., a quantizer in a digital receiver, some information about $S$ is lost. $\loss[S]{X\to Y}$ does not quantify the information (about $S$) lost over the channel, but the \emph{additional} information lost by quantizing the channel's output.
 
We now make a connection between Definition~\ref{def:relevantLoss} and the KLDR between $\Xvec$ and the $\Pvec$-lifted chain, $\Xvecg'^{\Pvec}$. To this end, let
\begin{equation}
 g^\bullet := \argmin_g \kldr{\Xvec}{\Xvecg'^{\Pvec}}.
\end{equation}
Recall from~\eqref{eq:diffbounds} that
\begin{equation}
 \kldr{\Xvec}{\Xvecg'^{\Pvec}}=\ent{Y_{g,n}|Y_{g,n-1}}-\ent{Y_{g,n}|X_{n-1}}
\end{equation}
which, by adding and subtracting $\ent{Y_{g,n}}$ can be rewritten as
\ifthenelse{\single=1}{
\begin{equation}
 \kldr{\Xvec}{\Xvecg'^{\Pvec}}=\mutinf{Y_{g,n};X_{n-1}}-\mutinf{Y_{g,n};Y_{g,n-1}}
=\loss[Y_{g,n}]{X_{n-1}\to Y_{g,n-1}}\label{eq:loss}
\end{equation}
}
{
\begin{align}
 \kldr{\Xvec}{\Xvecg'^{\Pvec}}&=\mutinf{Y_{g,n};X_{n-1}}-\mutinf{Y_{g,n};Y_{g,n-1}}\notag\\
&=\loss[Y_{g,n}]{X_{n-1}\to Y_{g,n-1}}\label{eq:loss}
\end{align}
}
where $\loss[Y_{g,n}]{X_{n-1}\to Y_{g,n-1}}$ is the \emph{information loss relevant w.r.t. $Y_{g,n}$} induced by projecting $X_{n-1}$ through the function $g$. Finding the optimal function $g^\bullet$ thus amounts to \emph{minimizing} information loss.

To the present date, we could not verify if this cost function has an interpretation in spectral theory. However, as mentioned above, it minimizes the difference between first-order upper and lower bounds on the entropy rate of the projected process $\Yvecg$ and, with~\cite[Thm.~9]{GeigerTemmel_kLump}, makes the projected process ``as Markov as possible''. 

\section{The Information Bottleneck Method: A Possible Way to Model Reduction}\label{ssec:ibmethod}
In this section we show that the state space reduction problem can be solved by a well-known information-theoretic algorithm: the information bottleneck method~\cite{Tishby_InformationBottleneck}. Since in~\cite{Geiger_Relevant_arXiv} the information bottleneck (IB) method was reformulated in terms of relevant information loss, the results of Section~\ref{ssec:plift} are essential for the development of the following paragraphs.

Let $X$ be a discrete RV representing an observation (e.g., the output of a noisy communications channel) or a data set. We are interested in a compressed representation $Y$ of this RV. In rate-distortion theory (e.g.,~\cite{Gray_Entropy}) one pursues the goal to minimize the mutual information between $X$ and its compression $Y$, $\mutinf{X;Y}$, subject to satisfying a certain distortion criterion $d$ (e.g., the mean-squared reconstruction error). This can be cast as a variational problem
\begin{equation}
\argmin_{p_{Y|X}} \mutinf{X;Y}+\beta d(X,Y)
\end{equation}
where $\beta$ is a Lagrange multiplier, and where stochastic compressions $p_{Y|X}(x,y)=\Prob{Y=y|X=x}$ are permitted.

The IB method takes up this approach by replacing the distortion measure by the negative mutual information between the compressed RV $Y$ and a \emph{relevant} RV $S$, representing the information one considers as meaningful and one wants to preserve (e.g., the binary input to the communications channel). The IB method therefore tries to solve
\begin{equation}
\argmin_{p_{Y|X}}\ \mutinf{X;Y}-\beta\mutinf{S;Y}\label{eq:IB}
\end{equation}
where the minimization runs over all stochastic relationships and where $\beta$ trades compression and preservation of information: A large value of $\beta$ places emphasis on preservation of relevant information, while a small value leads to high compression. Typical applications of the IB method include word and document clustering~\cite{Slonim_Document,Slonim_AgglomIB} or speech processing~\cite{Wohlmayr_Speech,Hecht_Speech}.

With $\beta\to\infty$, we focus on the second term of~\eqref{eq:IB} which can be rewritten with Definition~\ref{def:relevantLoss} as 
\begin{IEEEeqnarray}{RCL}
 \mutinf{S;Y} &=& \mutinf{S;X}-\loss[S]{X\to Y}.
\end{IEEEeqnarray}
With the restriction to deterministic compressions $p_{Y|X}$ determined by functions $g{:}\ \dom{X}\to\dom{Y}$, one obtains a formulation of the IB method which minimizes the relevant information loss, i.e., which solves
\begin{equation}
 \argmin_{g\in[\dom{X}\to\dom{Y}]}\ \loss[S]{X\to Y}.\label{eq:IBhard}
\end{equation}

For this problem, in~\cite{Slonim_AgglomIB} an iterative procedure, called \emph{agglomerative IB} (AIB) was introduced, which successively merges two elements of a partition of $\dom{X}$ until the desired cardinality $M$ is reached. The method is greedy, i.e., it minimizes the information lost in each step~\cite{Slonim_AgglomIB}, but does not guarantee that the global optimum~\eqref{eq:IBhard}, i.e., the least possible relevant information loss, is achieved.

Comparing~\eqref{eq:IBhard} with~\eqref{eq:loss}, one can see that the relevant information $Y_{g,n}$ depends on $g$, i.e., on the object to be optimized. Since in such a case the IB method is not applicable directly, we relax the problem by applying~\cite[Cor.~1]{Geiger_Relevant_arXiv}:
\begin{equation}
 \loss[Y_{g,n}]{X_{n-1}\to Y_{g,n-1}} \leq \loss[X_n]{X_{n-1}\to Y_{g,n-1}}\label{eq:lossDPI}
\end{equation}
Instead of minimizing $\kldr{\Xvec}{\Xvecg'^{\Pvec}}$, we only minimize its upper bound given by~\eqref{eq:lossDPI}. We thus look for
\begin{equation}
 g^{IB}:=\argmin_g  \loss[X_n]{X_{n-1}\to Y_{g,n-1}}.\label{eq:probIB}
\end{equation}
The possibility to apply the IB method and its algorithms (e.g., AIB) for state space reduction comes at the cost of optimality. As the next subsection shows, this cost is not as high as one would expect.


\subsection{Sub-Optimality of IB}\label{ssec:suboptimal}
By relaxing the $M$-partition problem to~\eqref{eq:probIB}, one loses the property that the cost function minimizes the upper bound on the aggregation error $\kldr{\Yvecg}{\Yvecg'}$. However, the obtained upper bound is still better than $\kldr{\Xvec}{\Xvecg'^{\muvec}}$:
\begin{align}
  & \loss[X_n]{X_{n-1}\to Y_{g,n-1}} \notag\\&= \ent{X_n|Y_{g,n-1}}-\ent{X_n|X_{n-1}}\\
&= \ent{X_n,Y_{g,n}|Y_{g,n-1}}-\entrate{X}\\
&= \ent{X_n|Y_{g,n},Y_{g,n-1}}+\underbrace{\ent{Y_{g,n}|Y_{g,n-1}}}_{=\entr{\Yvecg'}}-\entrate{X}\\
&\leq \ent{X_n|Y_{g,n}}+\entr{\Yvecg'}-\entrate{X}\\
&=\ent{X}-\ent{Y'_g}+\entr{\Yvecg'}-\entrate{X}\\
&=\redrate{X}-\redr{\Yvecg'}\\
&=\kldr{\Xvec}{\Xvecg'^{\muvec}}
\end{align}
where the last line is due to~\cite[Lem.~3]{Meyn_MarkovAggregation}. 

The solution of the relaxed problem~\eqref{eq:probIB} might not coincide with the solution of~\eqref{eq:loss}. To be specific: Even if a Markov chain $\Xvec$ is lumpable, neither the AIB nor the IB method implementing the relaxed optimization problem necessarily find the optimal $M$-partition. We will elaborate on this topic in the example in Section~\ref{ssec:AIBsuboptEx}.

\section{Examples}
\label{sec:examples}
In this section we illustrate our theoretical results at the hand of a few examples. In particular, we show the applicability of the information bottleneck method for Markov chain aggregation in Section~\ref{ssec:AIBexample}.

\subsection{Example 1}
We take the matrix given in~\cite[Section~V.A]{Meyn_MarkovAggregation}
\begin{equation}
 \Pvec=\left[\begin{array}{ccc}
              0.97&0.01&0.02\\0.02&0.48&0.50\\0.01&0.75&0.24
             \end{array}
\right]
\end{equation}
and use three different functions $g$ inducing the following partitions of $\dom{X}$: $\{\{1,2\},\{3\}\}$, $\{\{1,3\},\{2\}\}$, and $\{\{1\},\{2,3\}\}$.

For all the resulting aggregations, we compute upper bounds on the aggregation error using both the $\pivec$-lifting with $\pivec=\muvec$ and the $\Pvec$-lifting. In addition to that, the invariant distributions of the $\Pvec$-lifted Markov chains $\Xvecg'^{\Pvec}$ are computed and compared to $\muvec$, the invariant distribution of the original chain $\Xvec$. The results are shown in Table~\ref{tab:ex1}.

\begin{table}
\caption{Results of Example 1}
\label{tab:ex1}
 \begin{tabular}{|l||c||c|c|}
\hline
Partition & KLDR ($\muvec$) & KLDR ($\Pvec$) &  $\hat\muvec$ \\
& bit/sample & bit/sample & $[0.347,0.388,0.265]^T$\\
\hline
$\{\{1,2\},\{3\}\}$ & 0.823 & 0.185 & $[0.077,0.658,0.265]^T$\\
$\{\{1,3\},\{2\}\}$ & 0.808 & 0.317 & $[0.065,0.388,0.546]^T$\\
$\{\{1\},\{2,3\}\}$ & 0.037 & 0.001 & $[0.347,0.388,0.265]^T$\\
\hline
 \end{tabular}
\end{table}

As it can be seen, the partition $\{\{1\},\{2,3\}\}$ yields the best results in terms of KLDR. Moreover, it can be seen that the KLDR using $\Pvec$-lifting is smaller than the KLDR using $\pivec$-lifting in all three cases, as suggested by Theorem~\ref{thm:PliftingProperties}. However, unlike for $\pivec$-lifting with $\pivec=\muvec$, the invariant distribution obtained with our method depends on $g$ and in general differs from $\muvec$. An exception is the optimal partition, where $\Yvecg$ and $\Yvecg'$ are very close in terms of the KLDR, i.e., where $\Xvec$ is ``nearly'' lumpable w.r.t. $g$.

\subsection{Example 2} 
\label{ssec:AIBexample}
In this example we took the transition matrix $\Pvec$ from~\cite[Fig.~7]{Meyn_MarkovAggregation} and applied the agglomerative information bottleneck method~\cite{Slonim_AgglomIB} to aggregate the chain\footnote{We used the VLFeat Matlab implementation~\cite{VLFeat} of the agglomerative IB method.}, as described in Section~\ref{ssec:ibmethod}. As it can be seen in Fig.~\ref{fig:matrix}, the partitions of the alphabet appear to be reasonable and, for $M=5$, coincide with the solution obtained in~\cite{Meyn_MarkovAggregation}. In essence, the aggregation reduces the alphabet to groups of strongly interacting states.

An interesting fact can be observed by looking at Fig.~\ref{fig:curves}, which compares the KLDR curves for both lifting methods (the aggregation was obtained using the agglomerative IB method in both cases). While for $\pivec$-lifting the KLDR seems to be a function decreasing with increasing $M$, the same does not hold for $\Pvec$-lifting: If a certain partition is ``nearly'' lumpable, the KLDR curve exhibits a local minimum (cf. Theorem~\ref{thm:PliftingProperties}). Trivially, global minima with value zero are obtained for $M=1$ and $M=N$; thus, the curve depicted in Fig.~\ref{fig:curves} will decrease eventually if $M$ is further increased.

These results are relevant for properly choosing the cardinality of the reduced state space: For $\pivec$-lifting, it was suggested that a change in slope of the KLDR indicates that a meaningful partition was obtained~\cite[Section~V.D]{Meyn_MarkovAggregation}. Utilizing the tighter bound from $\Pvec$-lifting allows to choose the cardinality by detecting local minima.

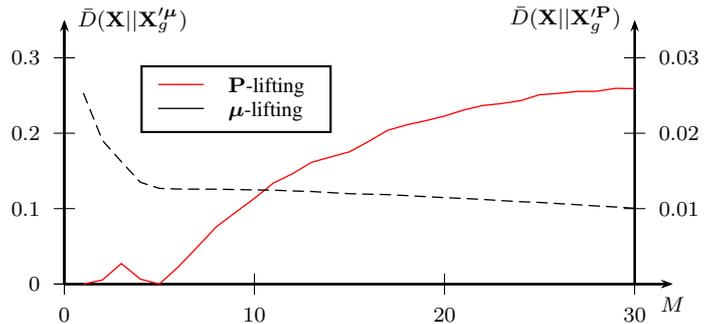
\begin{figure}[t]
 \centering
\begin{pspicture}[showgrid=false](-0.5,-0.5)(8.5,3.6)
	\footnotesize
	\psaxes[Dx=10,dx=2.5,Dy=0.1,dy=1]{->}(0.0,0)(8,3.5)[$M$,-90][$\kldr{\Xvec}{\Xvecg'^{\muvec}}$,0]
	\psaxes[labels=y,xAxis=false,ylabelPos=right,Dy=0.01,dy=1]{->}(7.5,0)(0,0)(7.5,3.5)[,0][$\kldr{\Xvec}{\Xvecg'^{\Pvec}}$,180]
	\rput[lb](1,2){\psframebox%
	{\begin{tabular}{ll}
	  	\psline[linewidth=0.5pt,linecolor=red](0.1,0.1)(0.7,0.1) &\hspace*{0.5cm} $\Pvec$-lifting\\%
		\psline[linewidth=0.5pt,linecolor=black](0.1,0.1)(0.7,0.1) &\hspace*{0.5cm} $\muvec$-lifting
	 \end{tabular}}}
	\readdata{\Mine}{MineKLDR.dat}
	\readdata{\Meyn}{MeynKLDR.dat}
	\psset{xunit=2.5mm,yunit=100cm}
	\dataplot[plotstyle=line,linecolor=red,linewidth=0.5pt]{\Mine}
    \psset{xunit=2.5mm,yunit=10cm}
	\dataplot[plotstyle=line,linecolor=black,style=Dash,linewidth=0.5pt]{\Meyn}
\end{pspicture}
 \caption{KLDR for the $\Pvec$- and the $\pivec$-lifting with $\pivec=\muvec$ ($\muvec$-lifting in the figure) for different cardinalities $M$ of the aggregated chain's alphabet. Both curves were obtained using the agglomerative IB method. Note that the KLDRs according to the different liftings are displayed with different scales, and note that the graph shows only $M\le 30<100=N$.}
 \label{fig:curves}
\end{figure}


\begin{figure*}
\centering
\subfigure[]{\includegraphics[width=0.3\textwidth]{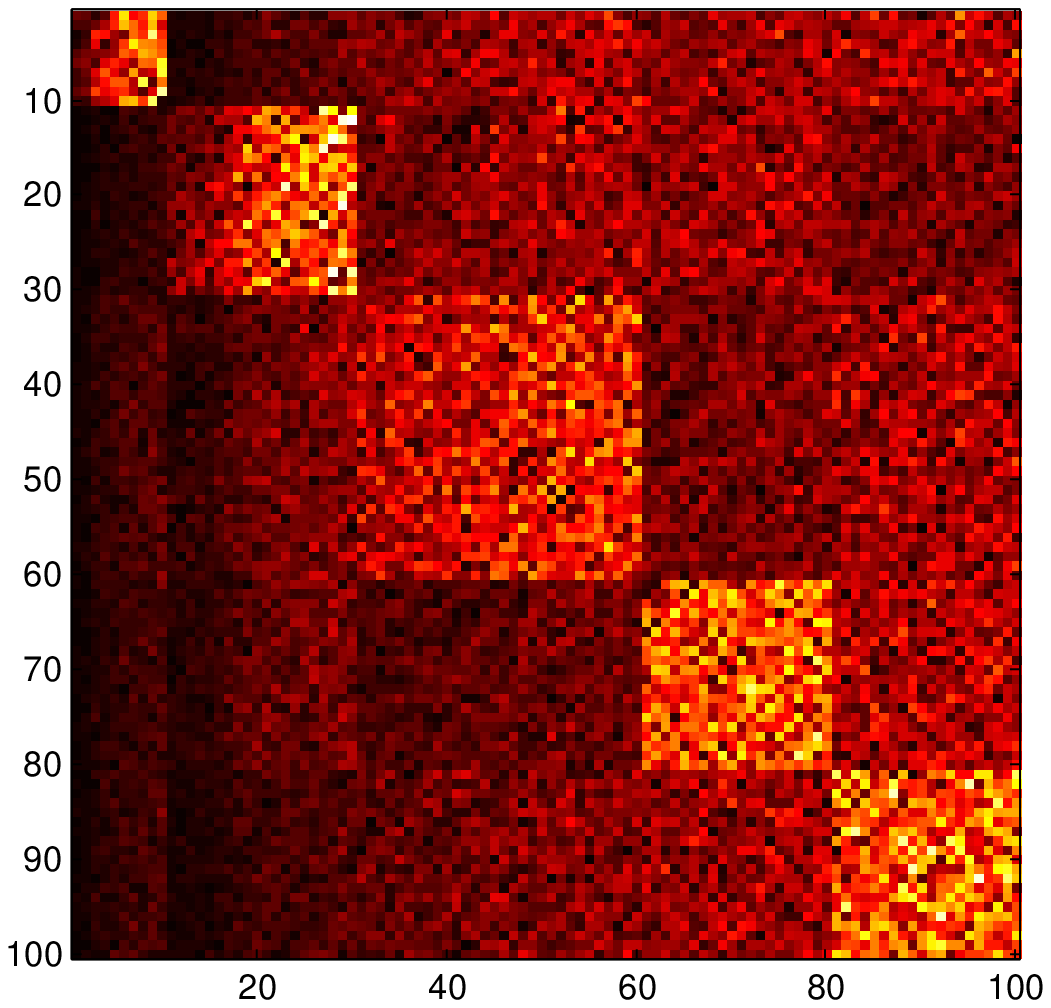}}\hfill
\subfigure[]{\includegraphics[width=0.3\textwidth]{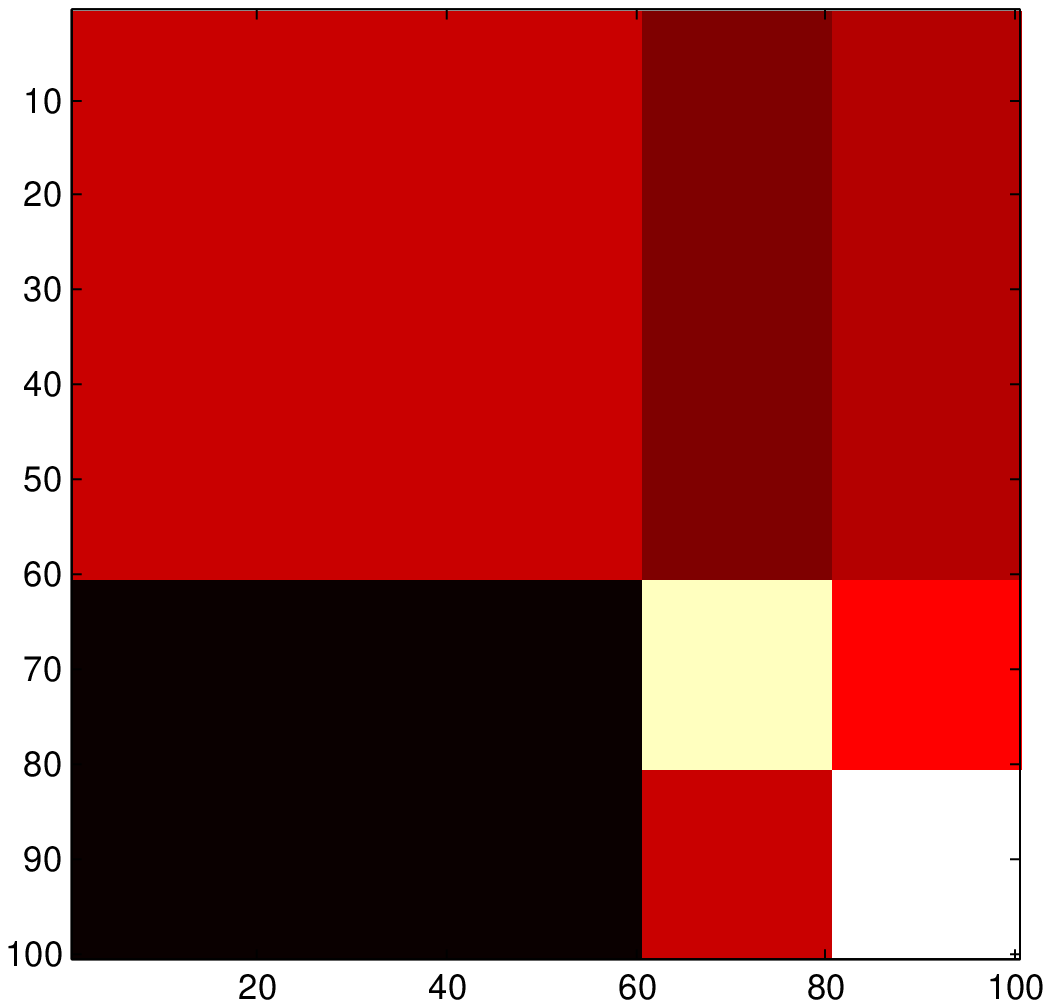}}\hfill
\subfigure[]{\includegraphics[width=0.3\textwidth]{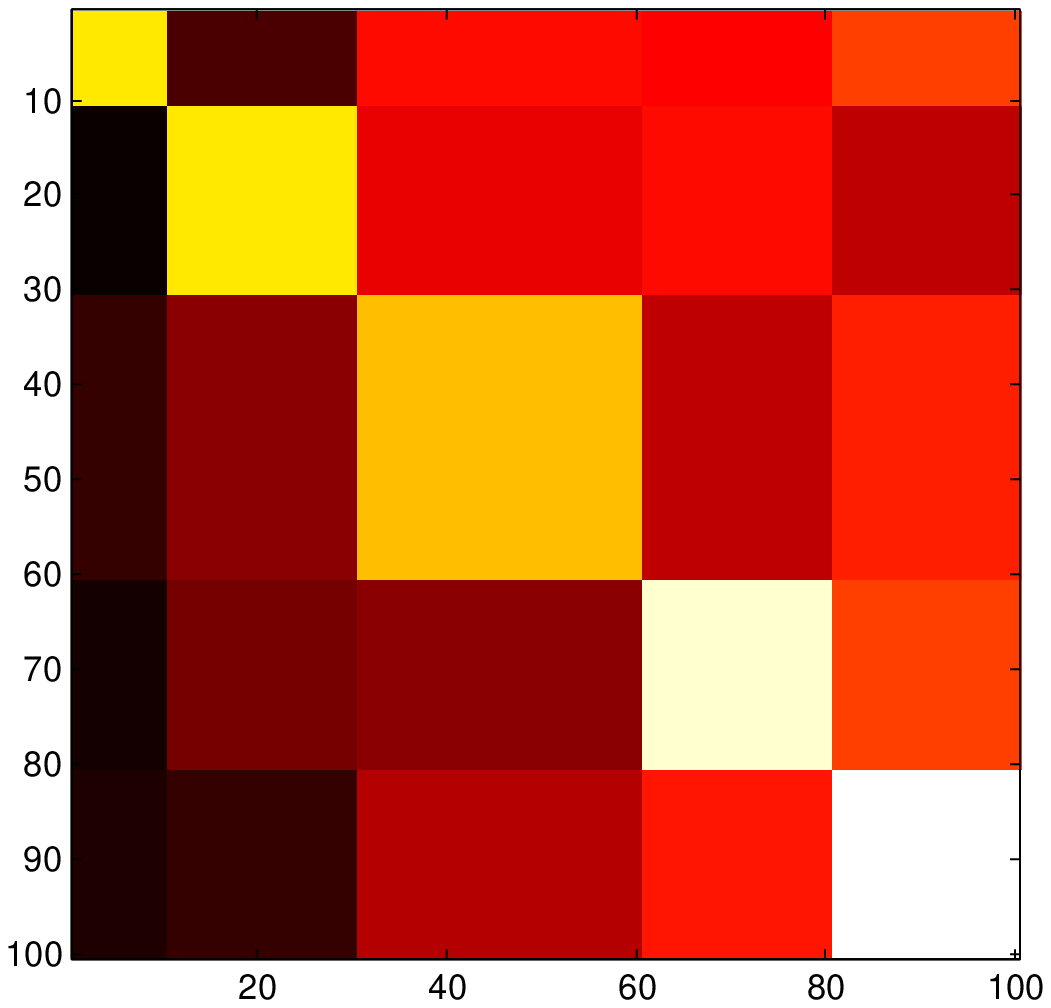}}
\caption{An illustration of Example 2: The original transition matrix (a) and the partitions obtained by using the agglomerative information bottleneck method. Blocks of the same color indicate that the corresponding states are mapped to the same output. $M=3$ (b) and $M=5$ (c).}
\label{fig:matrix}
\end{figure*}

\subsection{Example 3}
\label{ssec:AIBsuboptEx}
In this example we show that the relaxed optimization problem does not necessarily find the optimal partition. We start with a Markov chain $\Xvec$ with state space $\dom{X}=\{1,2,3\}$ and investigate the bi-partition problem (i.e., $M=2$). Let the transition matrix be given as
\begin{equation}
 \Pvec = \left[\begin{array}{ccc}
                0.0475 & 0.9025 & 0.05\\
   0.9025 & 0.0475 & 0.05\\
   0.95 & 0.05 & 0
               \end{array}
\right].
\end{equation}
Since this chain is lumpable for the partition $\{\{1, 2\},\{3\}\}$ (induced by the optimal function $g^\bullet$), one obtains
\begin{equation}
 \loss[Y_{g^\bullet,n}]{X_{n-1}\to Y_{g^\bullet,n-1}}=0.
\end{equation}
Computer simulations show, however, that this partition leads to a larger value of $\ent{X_n|Y_{g^\bullet,n-1}}$ than the other two options (namely, 1.19 bit compared to 0.55 and 0.69 bit, respectively). Since here the AIB and IB methods coincide\footnote{The bi-partition is obtained by merging two states.}, this example shows that the relaxation of the optimization problem does not necessarily lead to the optimal partition.

It is interesting to observe, however, that the information bottleneck method provides the same partition function as the method introduced in~\cite{Meyn_MarkovAggregation}, namely $\{\{1\},\{2,3\}\}$. The eigenvalues of the additive reversibilization of $\Pvec$ are $\lambda_1=1$, $\lambda_2=-0.038$, and $\lambda_3=-0.867$, the latter two inducing the partitions $\{\{1, 2\},\{3\}\}$ and $\{\{1\},\{2,3\}\}$, respectively. Hence, IB and the method in~\cite{Meyn_MarkovAggregation} respond with the solution related to the eigenvalue with the second-largest modulus, while the optimal solution remains to be related to the second-largest eigenvalue. This suggests a closer investigation of the interplay between the proposed cost function, its relaxation, and spectral theory, especially when the relevant eigenvalues are negative, cf.~\cite{Meyn_MarkovAggregation}.

\subsection{Example 4}
\label{ex:weaklyOnly}
We finally take an example from~\cite[pp.~139]{Kemeny_FMC}, which shows that our upper bound on the aggregation error is not tight in general, but only for lumpable $\Xvec$. To this end, let
\begin{equation}
 \Pvec=\left[\begin{array}{ccc}
\frac{1}{4} &\frac{1}{4} &\frac{1}{2} \\
0 &\frac{1}{6} &\frac{5}{6} \\
\frac{7}{8} &\frac{1}{8} &0
     \end{array}\right].
\end{equation}
As it can be verified easily, this chain is lumpable in the weak sense w.r.t. the partition $\{\{1\},\{2,3\}\}$, but not lumpable; i.e., $\Yvecg$ is a Markov chain if $\Xvec$ is initialized with the invariant distribution, but~\eqref{eq:condstrong} is not fulfilled. To show that the bound is not tight, we observe that $\kldr{\Yvecg}{\Yvecg'}=0$ but that, with
\begin{equation}
  \hat{\Pvec}=\left[\begin{array}{ccc}
\frac{1}{4} &\frac{1}{4} &\frac{1}{2} \\
\frac{7}{12} & \frac{5}{72} &\frac{25}{72} \\
\frac{7}{12} &\frac{5}{12} &0
     \end{array}\right] \neq \Pvec
\end{equation}
we get $\kldr{\Xvec}{\Xvecg'^{\Pvec}}=0.347>0$.

\def\Ivec{\mathbf{I}}
\def\tr{\sigma} 
\def\Z{\mathbf{Z}}
\def\st{i}
\def\stSet{\dom{X}}
\def\intervals{\mathbb{I}}
\def\cyl{{\cal C}}
\def\reaction{R}
\def\tracesSet{{\Sigma}^{[0,T]}}
\def\gene{G}
\def\gene{G}
\def\protein{P}
\def\Rvec{\mathbf{R}}

\section{Application to models of bio-molecular systems}
\label{sec:bioexample}

Recent advances in measurement techniques brought the need for quantitative modeling in biology 
\cite{Wilkinson_SystemsBiology}. 
Markov models are a major tool used for modeling the stochastic nature of bio-molecular interactions in cells. 
However, even the simplest networks with only a few interacting species can result in very large Markov chains, in which case their analysis becomes computationally inefficient or prohibitive. 
In these cases, reducing the state space of the model, with minimal information loss, is an important challenge. 
We illustrate on an example how our reduction method can be used in such a scenario.
The model defined by stochastic chemical kinetics evolves in continuous-time, following a  continuous-time Markov chain (CTMC). 
We will show how our aggregation method can be applied to reduce this CTMC by aggregating a subordinated DTMC.
The existing theory confirms that the resulting partition will also be suitable for the original CTMC.


For a well-mixed reaction system with molecular species $\spec{1},\ldots,\spec{\specN}$, the state of a system is typically modeled by a multiset of species' abundances: 
$\xvec=(x_1,\ldots,x_{\specN})\in \specStateSpace \subseteq \specSpace_{0}$. The dynamics of such a system are determined by a set of $\reacN$ reactions. The $k$-th reaction reads 
\begin{align}
\cons{1k}\spec{1},\ldots,\cons{\specN k}\spec{\specN k} \rA{\srate{k}} \prdn{1k}\spec{1},\ldots, \prdn{\specN k}\spec{\specN k},
\end{align}
where $\cons{ik}\in \N_{0}$ and $\prdn{ik}\in \N_{0}$ denote the substrate and product stoichiometric coefficients of species $i$, respectively, and where $c_k$ is the rate with which the reaction occurs.
If the $k$-th reaction occurs, after being in the state $\xvec$, the next state will be $\xvec+(\prdnvec{k}-\consvec{k}) = \xvec+\gaintvec{k}$, where $\gaintvec{k}$ is referred to as the stoichiometric change vector. 
The species multiplicities follow a continuous-time Markov chain and we denote the state of the system as the $t$-indexed random vector $\Xvec(t)=(X_1(t),\dots,X_{\specN}(t) )$. The probability of moving to the state $\xvec +\gaintvec{k}$ from $\xvec$ after time $\deltaT$ is 
$\Prob{\Xvec(t+\deltaT)=\xvec+\gaintvec{k}|\Xvec(t)=\xvec} = \sprop{k}(\xvec)\deltaT +o(\deltaT)$,
with $\sprop{k}$ the propensity of reaction $k$, the functional form of which is assumed to follow the principle of mass-action
$\sprop{k}(\xvec)= \srate{k}\prod_{i=1}^{\specN}{x_i\choose\cons{ik}}$ \cite{gillespie_07}.
%
The generator matrix $\Rvec{:}\ \specStateSpace \times \specStateSpace \rightarrow \R$ of the CTMC is determined by
$\Rvec(\xvec,\xvec+\gaintvec{k}) =\sprop{k}(\xvec)$, $\Rvec(\xvec,\xvec) = - \sum_{k=1}^{\reacN}\Rvec(\xvec,\xvec+\gaintvec{k})$, and zero otherwise.  

To illustrate, assume that a gene $\gene$ spontaneously turns on and off at rates $\srate{1}$ and $\srate{2}$ respectively, and that it regulates the expression of protein $\protein$.
More precisely, whenever a gene is turned on, the protein is synthesized at a rate $\srate{3}$, such that  $\srate{1},\srate{2}\ll\srate{3}$, that is, the gene activation is slow relative to the rate of protein synthesis.
Such a system requires a stochastic model and it can be specified with the following set of reactions: 
\begin{align*}
& \gene_0 \xrightleftharpoons[\srate{2}]{\srate{1}} \gene_1,\;\;\;
\gene_1+\protein_0\xrightharpoonup{c_3}\gene_1+\protein, \;\;\;
\protein \xrightharpoonup{\srate{4}}\protein_{0}
\end{align*}
Here, $\protein_0$ is introduced to simplify the system so that the total number of proteins is $n_{\protein}$. This is arguably more realistic than the unlimited birth-death process, as $\protein_0$ represents the (limited) pool of amino-acid building blocks for the proteins.
Finally, the protein can spontaneously degrade at rate $\srate{4}$.

Since the Markov process assigned to the model of a biochemical reaction network evolves in continuous-time, we cannot directly apply our aggregation method to the CTMC model of a biochemical network. Instead, we aggregate the subordinated DTMC:

\begin{definition}[Subordinated DTMC] \rm
Let $\Xvec$ be a CTMC over the state space $\dom{X}$ with generator matrix $\Rvec$ and transient marginal distribution $\pivec$, such that $\pi_i(t)=\Prob{\Xvec(t)=i}$. 
For $\unifc \geq \sup_{i\in\dom{X}} |\Rvec_{ii}|$, let $\Pvec:=\Rvec/\unifc+\Ivec_{N}$ ($\Ivec_N$ is an identity matrix of dimension $N$).
The DTMC defined by $\Pvec$ is the subordinated process of $\Xvec$ with uniformization constant $\lambda$, denoted by $\Xvec_{\unifc}$.
\end{definition}

\begin{figure}
\centering
\subfigure[]{\includegraphics[width=0.5\textwidth]{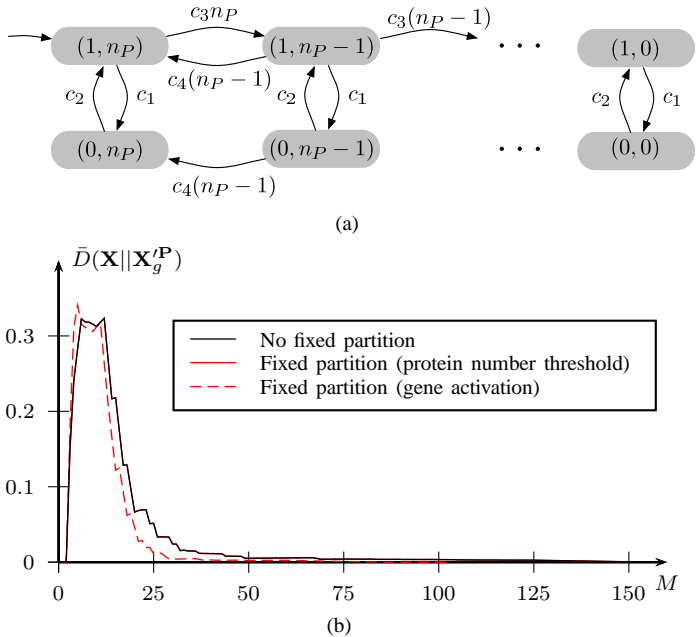}}\\
\subfigure[]{
\begin{pspicture}[showgrid=false](-0.5,-0.5)(8.,4)
	\footnotesize
	\psaxes[Dx=25,dx=1.25,Dy=0.1,dy=1]{->}(0.0,0)(8,4.0)[$M$,-90][$\kldr{\Xvec}{\Xvecg'^{\Pvec}}$,0]
	\rput[lb](1.5,2){\psframebox%
	{\begin{tabular}{ll}
	  	\psline[linewidth=0.5pt,linecolor=black](0.1,0.1)(0.7,0.1) &\hspace*{0.5cm} No fixed partition\\%
		\psline[linewidth=0.5pt,linecolor=red](0.1,0.1)(0.7,0.1) &\hspace*{0.5cm} Fixed partition (protein number threshold)\\
		\psline[linewidth=0.5pt,linecolor=red,style=Dash](0.1,0.1)(0.7,0.1) &\hspace*{0.5cm} Fixed partition (gene activation)\\%
	 \end{tabular}}}
	\readdata{\Full}{Full.dat}
	\readdata{\FixedLarge}{Fixed101.dat}
	\readdata{\FixedSmall}{Fixed20.dat}
	\psset{xunit=0.5mm,yunit=20cm}
	\listplot[xEnd=150,plotstyle=line,linecolor=red,style=Dash,linewidth=0.5pt]{\FixedLarge}
	\listplot[xEnd=150,plotstyle=line,linecolor=red,linewidth=0.5pt]{\FixedSmall}
	\listplot[xEnd=150,plotstyle=line,linecolor=black,linewidth=0.5pt]{\Full}
\end{pspicture}
}
\caption{Application to modeling bio-molecular interactions:
(a) The continuous-time Markov chain assigned to the gene expression example;
(b) The upper bound to the KLDR obtained by agglomerative information bottleneck method, for all partition sizes:
(black) no partition class is fixed, 
(red) a partition class where the number of proteins is bigger than a threshold $T=0.9n=180$ is fixed, 
(red, dashed) all states where the gene is turned on (a cluster with $101$ states) are fixed. 
The red line is not visible because it equals the black line. 
The aggregation error is displayed only up to a partition size of $M=150$.}
\label{fig:model2}
\end{figure}

The subordinated\footnote{The DTMC $\Xvec_{\unifc}$ is also called a \emph{uniformized} or \emph{randomized} chain.} DTMC agrees with the original process in its transient distribution~\cite{norris97}.
Moreover, the KLDR between the subordinated DTMCs equals to the KLDR between the original CTMCs in the limit of a large uniformization constant, as can be shown by discretizing the time domain (see~\cite{cohn2010mean} for detailed presentation).
The definition of the KLDR for CTMCs and its existence criterion can be found in \cite[Ch.~6]{tanja_thesis}.
In the algorithm, we choose the uniformization constant $\unifc = \sup_{i\in\dom{X}} |\Rvec_{ii}|+1$.

For the initial vector $\Xvec(0)=(1,n_{\protein})$ 
(where the components denote copy numbers of $\gene_1$ and $P$ respectively), 
the CTMC has $N=2(n_{\protein}+1)$ reachable states (Fig.~\ref{fig:model2}a). 
After the chain exhibits stationary behavior, the algorithm is applied for $n_{\protein}=100$ and  for $M=1,2,\ldots,202$. 
Moreover, the algorithm was adapted to search for the optimal partition after one partition class is fixed.
This is desirable in scenarios where the modeler a priori wants to track the joint probability of all the states that satisfy a certain property.
For example, one may be interested in a priori clustering those states for which the number of proteins is bigger than a given threshold $T=0.9n_{\protein}$, or all the states where the gene is turned on (depicted in the top row in Fig.~\ref{fig:model2}a).
In Fig.~\ref{fig:model2}b, we compare the upper bound on the aggregation error for the optimal partition, and 
optimal partition upon fixing each of the two mentioned partition classes. 
The results 
confirm that our algorithm provides only sub-optimal solutions, because, for example, lumping a priori all states with an activated gene yields a better bound than when no partition class is fixed.
In particular, notice that for $M=2$, lumping all states where the gene is turned on satisfies the criterion of lumpability, rendering 
the upper bound on the error to be within numerical precision. 

\section{Conclusion}
In this work we presented a new method for Markov chain state space reduction based on information-theoretic criteria. Specifically, the Kullback-Leibler divergence rate between the process obtained by simply partitioning the alphabet of the original chain and its best Markov approximation is employed as a cost function. The Kullback-Leibler divergence rate between the original chain and the \emph{lifting} of the optimal Markov approximation was shown to yield an upper bound on the cost function.

By properly defining the lifting, we not only obtain the best upper bound under certain restrictions, but also a cost function which links the reduced-alphabet model to the notion of lumpability. In addition to that, it is shown that the information bottleneck method can be used for model reduction by relaxing the optimization problem. Future work shall investigate possible connections between the proposed cost function and the spectral theory of Markov chains, the extension to non-stationary Markov chains, and the generalization to stochastic aggregations.

\appendices
\section{Proof of Proposition~\ref{prop:liftingProperties}}
\label{app:pivecProperties}
For the first property (which is also mentioned in~\cite{Meyn_MarkovAggregation}) we show that the condition
\begin{equation}
 \Vvec\Uvec^{\zetavec}\Pvec'\Vvec = \Pvec'\Vvec 
\end{equation}
from Lemma~\ref{lem:lump} holds for all possible $\pivec$-liftings and for all positive probability vectors $\zetavec$. Letting $\Pvec'=\Vvec\Qvec\Uvec^{\pivec}$,
\begin{IEEEeqnarray}{RCL}
 \Vvec\Uvec^{\zetavec}\Vvec\Qvec\Uvec^{\pivec}\Vvec = \Vvec\Qvec\Uvec^{\pivec}\Vvec.
\end{IEEEeqnarray}
Since $\Uvec^{\pivec}\Vvec=\eye$ for all positive probability vectors $\pivec$, equality is achieved and the first result is proved.

For the second property (cf.~\cite[Thm.~2, Property~3]{Meyn_MarkovAggregation}) note that with $\Pvec'=\Vvec\Qvec\Uvec^{\muvec}$,
\begin{equation}
 \muvec^T\Pvec' = \muvec^T\Vvec\Qvec\Uvec^{\muvec} = \nuvec^T\Qvec\Uvec^{\muvec} = \nuvec^T\Uvec^{\muvec}=\muvec^T
\end{equation}
where the second and third equality are due to Lemma~\ref{lem:lump} and the last follows from the definition of $\Uvec^{\muvec}$ in~\eqref{eq:umatdef}.
                                                                        
For the third property we refer the reader to~\cite[Thm.~1]{Meyn_MarkovAggregation}.

Next, observe that $Q_{g(i)g(j)}=0$ implies $P_{ij}=0$ (see~\eqref{eq:Qmatlong}, combined with the fact that $\muvec$ is positive~\cite[Thm.~4.1.4]{Kemeny_FMC}). For the entries of the lifted matrix we can now write
\begin{equation}
 P'_{ij} = \frac{\pi_j}{\sum_{k\in\preim{g(j)}} \pi_k} Q_{g(i)g(j)}.
\end{equation}
Since $\pivec$ is positive, it follows that $P'_{ij}=0$ implies $P_{ij}=0$, or equivalently, $\Pvec'\gg\Pvec$. 

The last property immediately follows from Lemma~\ref{lem:kldbound}; the lemma may be applied because $\Xvecg'^{\muvec}$ is lumpable to $\Yvecg'$ (property 1) and $\Pvec'\gg\Pvec$ (property 4). This completes the proof.\endproof
\section{Proof of Theorem~\ref{thm:PliftingProperties}} 
\label{app:PProperties}
For the proof we note that another condition for lumpability is given by the entries of the matrix $\mathbf{R}=\Pvec\Vvec$. In particular, iff for all $h,l\in\dom{Y}$ the elements
\begin{equation}
 R_{il}:=\sum_{j\in\preim{l}} P_{ij} \label{eq:KemenyLump}
\end{equation}
are the same for all $i\in\preim{h}$, the chain is lumpable w.r.t. $g$~\cite[Thm.~6.3.2]{Kemeny_FMC}. Using this with~\eqref{eq:myLift} one gets
\begin{IEEEeqnarray}{RCL}
 \hat{R}_{il} &=& \sum_{j\in\preim{l}} \hat{P}_{ij} = Q_{hl}.
\end{IEEEeqnarray}
Clearly, $\hat{R}_{il}$ assumes the same values for all $i\in\preim{h}$, as required. This completes the proof of the first statement\footnote{Note that this statement holds for all stochastic matrices used for lifting, i.e., the lifting matrix does not have to be equal to the transition matrix of the original chain.}.

The second statement is obvious from the Definition of $\Pvec$-lifting and from the proof of property 4 of Proposition~\ref{prop:liftingProperties}.

For the third statement, we introduce an arbitrary lifting
\begin{equation}
 \tilde{P}_{ij} = b_{ij}Q_{g(i)g(j)}
\end{equation}
subject to $\sum_{j\in\preim{l}}b_{ij}=1$ for all $l\in\dom{Y}$ and all $i\in\dom{X}$. With~\eqref{eq:KemenyLump}, this condition is necessary and sufficient for lumpability of the lifted chain $\tilde{\Xvec}$ with transition matrix $\tilde{\Pvec}$. We write for the KLDR
\begin{IEEEeqnarray}{RCL}
 \kldrate{X}{\tilde{X}} &=& \sum_{i,j\in\dom{X}}\mu_iP_{ij}\log\frac{P_{ij}}{\tilde{P}_{ij}}\\
&=& \sum_{i,j\in\dom{X}}\mu_iP_{ij}\log\frac{P_{ij}}{ b_{ij}Q_{g(i)g(j)}}\\
&=& \entrate{\Yvecg'}-\entrate{X}+\sum_{i,j\in\dom{X}}\mu_iP_{ij}\log\frac{1}{ b_{ij}}.
\end{IEEEeqnarray}
The last term can be written as
\begin{IEEEeqnarray*}{RCL}
 \sum_{i,j\in\dom{X}}\mu_iP_{ij}\log\frac{1}{ b_{ij}} &=& \sum_{i\in\dom{X}}\mu_i\sum_{l\in\dom{Y}}R_{il}\sum_{j\in\preim{l}}\frac{P_{ij}}{R_{il}}\log\frac{1}{ b_{ij}}.
\end{IEEEeqnarray*}
Here, the last term on the right is a cross-entropy, since both $b_{ij}$ and $\frac{P_{ij}}{R_{il}}$ are probability vectors on $\preim{l}$. The cross-entropy is minimized\footnote{This is a direct consequence of the fact that the Kullback-Leibler divergence vanishes if and only if the considered probability mass functions are equal~\cite[pp.~31]{Cover_Information2}.} iff for all $j\in\preim{l}$
\begin{equation}
  b_{ij}=\frac{P_{ij}}{R_{il}}=\frac{P_{ij}}{\sum_{k\in\preim{l}}P_{ik}}.
\end{equation}
Since the sums over $i$ and $l$ are expecations, the minimum is achieved iff above condition holds also for all $i\in\dom{X}$ and all $l\in\dom{Y}$ for which $R_{il}>0$. If $R_{il}=0$, the assignment for $b_{ij}$ is immaterial for $j\in\preim{l}$.

To show that the $\Pvec$-lifting indeed yields a better bound observe that with $\ent{Y_{g,n}|Y_{g,n-1}}=\entrate{\Yvecg'}$
\begin{align}
&\kldr{\Xvec}{\Xvecg'^{\muvec}}-\kldr{\Xvec}{\Xvecg'^{\Pvec}}\notag\\
&=\ent{X}-\entrate{X}-\ent{Y_g'}+\ent{Y_{g,n}|X_{n-1}}\\
&=\ent{X_n}-\ent{X_n|X_{n-1}}-\ent{Y_{g,n}}+\ent{Y_{g,n}|X_{n-1}}\\
&=\mutinf{X_n;X_{n-1}}-\mutinf{Y_{g,n};X_{n-1}}\\
&\geq0
\end{align}
by the data processing inequality. $\kldr{\Xvec}{\Xvecg'^{\Pvec}}\geq\kldr{\Yvecg}{\Yvecg'}$ is obtained by Lemma~\ref{lem:kldbound}, see Proposition~\ref{prop:liftingProperties}, property 5).

For the fifth property, note that the sufficient and necessary condition for lumpability~\eqref{eq:condstrong}, namely that 
\begin{equation}
 R_{il}=\sum_{j\in\preim{l}} P_{ij}=Q_{hl}
\end{equation}
is the same for all $i\in\preim{h}$, can be used in the definition of $\hat{\Pvec}$:
\begin{equation}
 \hat{P}_{ij}=\frac{P_{ij}}{\sum_{k\in\dom{S}_j}P_{ik}} Q_{g(i)g(j)}
= \frac{P_{ij}}{\sum_{k\in\dom{S}_j}P_{ik}} R_{ig(j)}=P_{ij}
\end{equation}
This proves the ``$\Rightarrow$'' part. The ``$\Leftarrow$'' part follows from Lemma~\ref{lem:kldbound}. This completes the proof.\endproof

\section*{Acknowledgments}
The authors wish to thank Kun Deng and Prashant Mehta for providing the data set for the example in Section~\ref{ssec:AIBexample}.

\bibliographystyle{IEEEtran}
\bibliography{IEEEabrv,lumping.bib}

\begin{thebibliography}{10}
\providecommand{\url}[1]{#1}
\csname url@samestyle\endcsname
\providecommand{\newblock}{\relax}
\providecommand{\bibinfo}[2]{#2}
\providecommand{\BIBentrySTDinterwordspacing}{\spaceskip=0pt\relax}
\providecommand{\BIBentryALTinterwordstretchfactor}{4}
\providecommand{\BIBentryALTinterwordspacing}{\spaceskip=\fontdimen2\font plus
\BIBentryALTinterwordstretchfactor\fontdimen3\font minus
  \fontdimen4\font\relax}
\providecommand{\BIBforeignlanguage}[2]{{%
\expandafter\ifx\csname l@#1\endcsname\relax
\typeout{** WARNING: IEEEtran.bst: No hyphenation pattern has been}%
\typeout{** loaded for the language `#1'. Using the pattern for}%
\typeout{** the default language instead.}%
\else
\language=\csname l@#1\endcsname
\fi
#2}}
\providecommand{\BIBdecl}{\relax}
\BIBdecl

\bibitem{Wilkinson_SystemsBiology}
D.~Wilkinson, \emph{Stochastic Modelling for Systems Biology}, ser. Chapman \&
  Hall/CRC Mathematical \& Computational Biology.\hskip 1em plus 0.5em minus
  0.4em\relax Boca Raton, FL: Taylor \& Francis, 2006.

\bibitem{Manning_NLP}
C.~D. Manning and H.~Sch\"utze, \emph{Foundations of Statistical Natural
  Language Processing}, 2nd~ed.\hskip 1em plus 0.5em minus 0.4em\relax
  Cambridge: MIT Press, 2000.

\bibitem{Meyn_MarkovAggregation}
K.~Deng, P.~G. Mehta, and S.~P. Meyn, ``Optimal {Kullback}-{Leibler}
  aggregation via spectral theory of {Markov} chains,'' \emph{{IEEE} Trans.
  Autom. Control}, vol.~56, no.~12, pp. 2793--2808, Dec. 2011.

\bibitem{Aldhaheri_NCDMC}
R.~W. Aldhaheri and H.~K. Khalil, ``Aggregation of the policy iteration method
  for nearly completely decomposable {Markov} chains,'' \emph{{IEEE} Trans.
  Autom. Control}, vol.~36, no.~2, pp. 178--187, Feb. 1991.

\bibitem{Rached_KLDR}
Z.~Rached, F.~Alajaji, and L.~L. Campbell, ``The {Kullback-Leibler} divergence
  rate between {Markov} sources,'' \emph{{IEEE} Trans. Inf. Theory}, vol.~50,
  no.~5, pp. 917--921, May 2004.

\bibitem{Tishby_InformationBottleneck}
N.~Tishby, F.~C. Pereira, and W.~Bialek, ``The information bottleneck method,''
  in \emph{Proc. Allerton Conf. on Communication, Control, and Computing}, Sep.
  1999, pp. 368--377.

\bibitem{White_HMMLumpable}
L.~B. White, R.~Mahony, and G.~D. Brushe, ``Lumpable hidden {Markov}
  models--model reduction and reduced complexity filtering,'' \emph{{IEEE}
  Trans. Autom. Control}, vol.~45, no.~12, pp. 2297--2306, Dec. 2000.

\bibitem{Jia_MDP}
Q.-S. Jia, ``On state aggregation to approximate complex value functions in
  large-scale {Markov} decision processes,'' \emph{{IEEE} Trans. Autom.
  Control}, vol.~56, no.~2, pp. 333--344, Feb. 2011.

\bibitem{Vidyasagar_MarkovAgg}
M.~Vidyasagar, ``Reduced-order modeling of {Markov} and hidden {Markov}
  processes via aggregation,'' in \emph{Proc. IEEE Conf. on Decision and
  Control (CDC)}, Atlanta, Dec. 2010, pp. 1810--1815.

\bibitem{Deng_LowRank}
K.~Deng and D.~Huang, ``Model reduction of {Markov} chains via low-rank
  approximation,'' in \emph{Proc. American Control Conf. (ACC)}, Montreal, Jun.
  2012, pp. 2651--2656.

\bibitem{Katsoulakis_CoarseGraining}
M.~A. Katsoulakis and J.~Trashorras, ``Information loss in coarse-graining of
  stochastic particle dynamics,'' \emph{Journal of statistical physics}, vol.
  122, no.~1, pp. 115--135, 2006.

\bibitem{Meila_Segmentation}
M.~Meil\u{a} and J.~Shi, ``Learning segmentation by random walks,'' in
  \emph{Advances in Neural Information Processing Systems (NIPS)}, Denver, CO,
  Nov. 2000, pp. 1--7.

\bibitem{Runolfsson_ModelReduction}
T.~Runolfsson and Y.~Ma, ``Model reduction of nonreversible {Markov} chains,''
  in \emph{Proc. IEEE Conf. on Decision and Control (CDC)}, New Orleans, LA,
  Dec. 2007, pp. 3739--3744.

\bibitem{Vidyasagar_Distribution}
M.~Vidyasagar, ``A metric between probability distributions on finite sets of
  different cardinalities and applications to order reduction,'' \emph{{IEEE}
  Trans. Autom. Control}, vol.~57, no.~10, pp. 2464--2477, Oct. 2012.

\bibitem{Geiger_Relevant_arXiv}
B.~C. Geiger and G.~Kubin, ``Signal enhancement as minimization of relevant
  information loss,'' in \emph{Proc. ITG Conf. on Systems, Communication and
  Coding}, Munich, Jan. 2013, pp. 1--6, extended version available: {\tt
  arXiv:1205.6935 [cs.IT]}.

\bibitem{Kemeny_FMC}
J.~G. Kemeny and J.~L. Snell, \emph{Finite Markov Chains}, 2nd~ed.\hskip 1em
  plus 0.5em minus 0.4em\relax Springer, 1976.

\bibitem{Raj_Clustering}
A.~Raj and C.~H. Wiggins, ``An information-theoretic derivation of
  min-cut-based clustering,'' \emph{{IEEE} Trans. Pattern Anal. Mach. Intell.},
  vol.~32, no.~6, pp. 988--995, Jun. 2010.

\bibitem{Tishby_MarkovRelax}
N.~Tishby and N.~Slonim, ``Data clustering by {Markovian} relaxation and the
  information bottleneck method,'' in \emph{Advances in Neural Information
  Processing Systems (NIPS)}, 2001.

\bibitem{Friedman_Clustering}
A.~Friedman and J.~Goldberger, ``Information theoretic pairwise clustering,''
  in \emph{SIMBAD}, ser. LNCS, E.~Hancock and M.~Pelillo, Eds.\hskip 1em plus
  0.5em minus 0.4em\relax Berlin: Springer, 2013, vol. 7953, pp. 106--119.

\bibitem{Slonim_AgglomIB}
N.~Slonim and N.~Tishby, ``Agglomerative information bottleneck,'' in
  \emph{Advances in Neural Information Processing Systems (NIPS)}.\hskip 1em
  plus 0.5em minus 0.4em\relax MIT Press, 1999, pp. 617--623.

\bibitem{Xu_DA}
Y.~Xu, S.~M. Salapaka, and C.~L. Beck, ``Aggregation of graph models and
  {Markov} chains by deterministic annealing,'' \emph{{IEEE} Trans. Autom.
  Control}, accepted for publication.

\bibitem{Kieffer_MarkovChannelsAMS}
J.~C. Kieffer and M.~Rahe, ``{Markov} channels are asymptotically mean
  stationary,'' \emph{Siam Journal of Mathematical Analysis}, vol.~12, pp.
  293--305, 1980.

\bibitem{Gray_Probability}
R.~M. Gray, \emph{Probability, Random Processes, and Ergodic Properties},
  2nd~ed.\hskip 1em plus 0.5em minus 0.4em\relax New York, NY: Springer, 2009.

\bibitem{Cover_Information2}
T.~M. Cover and J.~A. Thomas, \emph{Elements of Information Theory},
  2nd~ed.\hskip 1em plus 0.5em minus 0.4em\relax Hoboken, NJ: Wiley
  Interscience, 2006.

\bibitem{Gray_Entropy}
R.~M. Gray, \emph{Entropy and Information Theory}.\hskip 1em plus 0.5em minus
  0.4em\relax New York, NY: Springer, 1990.

\bibitem{Pinsker_InfoEngl}
M.~S. Pinsker, \emph{Information and Information Stability of Random Variables
  and Processes}.\hskip 1em plus 0.5em minus 0.4em\relax San Francisco, CA:
  Holden Day, 1964.

\bibitem{Watanabe_InfoLoss}
S.~Watanabe and C.~T. Abraham, ``Loss and recovery of information by coarse
  observation of stochastic chain,'' \emph{Information and Control}, vol.~3,
  no.~3, pp. 248--278, Sep. 1960.

\bibitem{Petrov_Lumpability}
J.~Feret, T.~Henzinger, H.~Koeppl, and T.~Petrov, ``Lumpability abstractions of
  rule-based systems,'' \emph{Theoretical Computer Science}, vol. 431, pp.
  137--164, 2012.

\bibitem{Weinan_Aggregation}
E.~Weinan, L.~Tiejun, and E.~Vanden-Eijnden, ``Optimal partition and effective
  dynamics of complex networks,'' \emph{PNAS}, vol. 105, no.~23, pp.
  7907--7912, Jun. 2008.

\bibitem{Vidyasagar_KLDRate}
M.~Vidyasagar, ``{Kullback}-{Leibler} divergence rate between probability
  distributions on sets of different cardinalities,'' in \emph{Proc. of the
  IEEE Conf. on Decision and Control}, Atlanta, GA, Dec. 2010, pp. 948--953.

\bibitem{GeigerTemmel_kLump}
B.~C. Geiger and C.~Temmel, ``Lumpings of {Markov} chains, entropy rate
  preservation, and higher-order lumpability,'' Dec. 2012, accepted in J. Appl.
  Prob.; preprint available: {\tt arXiv:1212.4375 [cs.IT]}.

\bibitem{Slonim_Document}
N.~Slonim and N.~Tishby, ``Document clustering using word clusters via the
  information bottleneck method,'' in \emph{Proc. of the Int. ACM SIGIR Conf.
  on Research and Development in Information Retrieval}.\hskip 1em plus 0.5em
  minus 0.4em\relax New York, NY, USA: ACM, 2000, pp. 208--215.

\bibitem{Wohlmayr_Speech}
M.~Wohlmayr, M.~Markaki, and Y.~Stylianou, ``Speech-nonspeech discrimination
  based on speech-relevant spectrogram modulations,'' in \emph{Proc. European
  Signal Processing Conf. (EUSIPCO)}, Poznan, Sep. 2007, pp. 1551 -- 1555.

\bibitem{Hecht_Speech}
R.~M. Hecht, E.~Noor, and N.~Tishby, ``Speaker recognition by {Gaussian}
  information bottleneck,'' in \emph{Proc. InterSpeech}, Brighton, Sep. 2009,
  pp. 1567--1570.

\bibitem{VLFeat}
A.~Vedaldi and B.~Fulkerson, ``{VLFeat}: An open and portable library of
  computer vision algorithms,'' \url{http://www.vlfeat.org/}, 2008.

\bibitem{gillespie_07}
D.~T. Gillespie, ``{Stochastic Simulation of Chemical Kinetics},'' \emph{Annual
  Review of Physical Chemistry}, vol.~58, no.~1, pp. 35--55, 2007.

\bibitem{norris97}
J.~R. Norris, \emph{Markov chains}.\hskip 1em plus 0.5em minus 0.4em\relax
  Cambridge university press, 1998, no. 2008.

\bibitem{cohn2010mean}
I.~Cohn, T.~El-Hay, N.~Friedman, and R.~Kupferman, ``Mean field variational
  approximation for continuous-time bayesian networks,'' \emph{The Journal of
  Machine Learning Research}, vol. 9999, pp. 2745--2783, 2010.

\bibitem{tanja_thesis}
T.~Petrov, ``Formal reductions of stochastic rule-based models of biochemical
  systems,'' Ph.D. dissertation, ETHZ, 2013.

\end{thebibliography}

\begin{IEEEbiography}[{\includegraphics[width=1in,height=1.25in,clip,keepaspectratio]{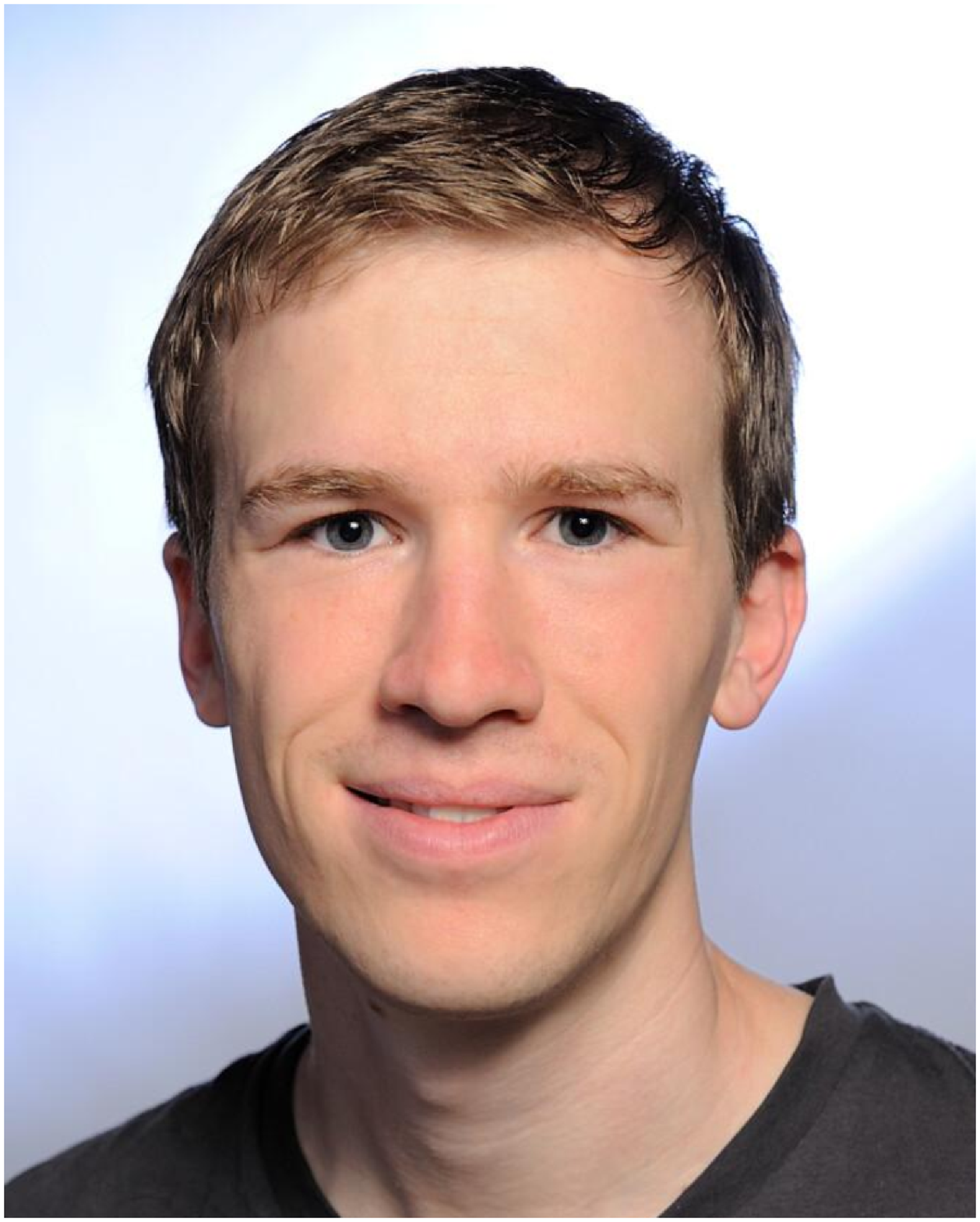}}]{Bernhard C. Geiger}
(S'07, M'14) was born in Graz, Austria, in 1984. He received the Dipl.-Ing. degree in electrical engineering (with distinction) and the Dr. techn. degree in electrical and information engineering (with distinction) from Graz University of Technology, Austria in 2009 and 2014, respectively.

In 2009 he joined the Signal Processing and Speech Communication Laboratory, Graz University of Technology, as a Project Assistant and took a position as a Research and Teaching Associate at the same lab in 2010. He is currently a postdoctoral researcher at the Institute of Communication Engineering, TU Munich. His research interests cover the intersection of information theory with system theory and signal processing, certain topics in the theory of Markov chains, and the analysis of GNSS acquisition.
\end{IEEEbiography}

\begin{IEEEbiography}[{\includegraphics[width=1in,height=1.25in,clip,keepaspectratio]{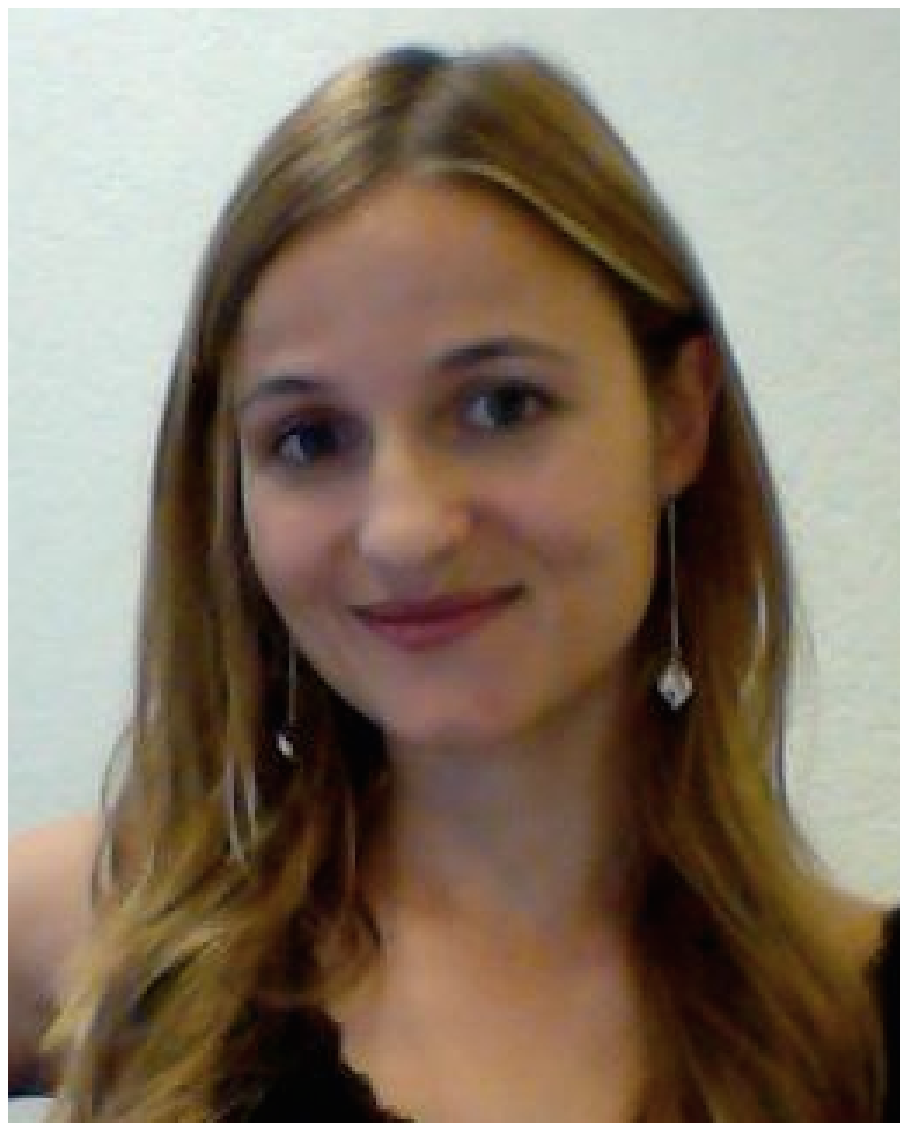}}]{Tatjana Petrov}
is a postdoctoral Fellow at IST Austria. She holds a Ph.D. degree from ETH Zurich (Swiss Federal School of Technology Zurich) and a M.Sc. degree in Theoretical Computer Science from University of Novi Sad, Serbia. In her PhD work, in the scope of SystemsX (the Swiss Initiative for Systems Biology), towards understanding complex signaling pathway dynamics, Tatjana developed a formal framework for exact and approximate reductions of stochastic rule-based models of complex biochemical networks. Tatjana’s current work focuses on formal verification approaches to modeling and analysis of complex systems, with a major interest in stochastic and hybrid systems, interface theories, and applications to molecular biology.
\end{IEEEbiography}

\begin{IEEEbiography}[{\includegraphics[width=1in,height=1.25in,clip,keepaspectratio]{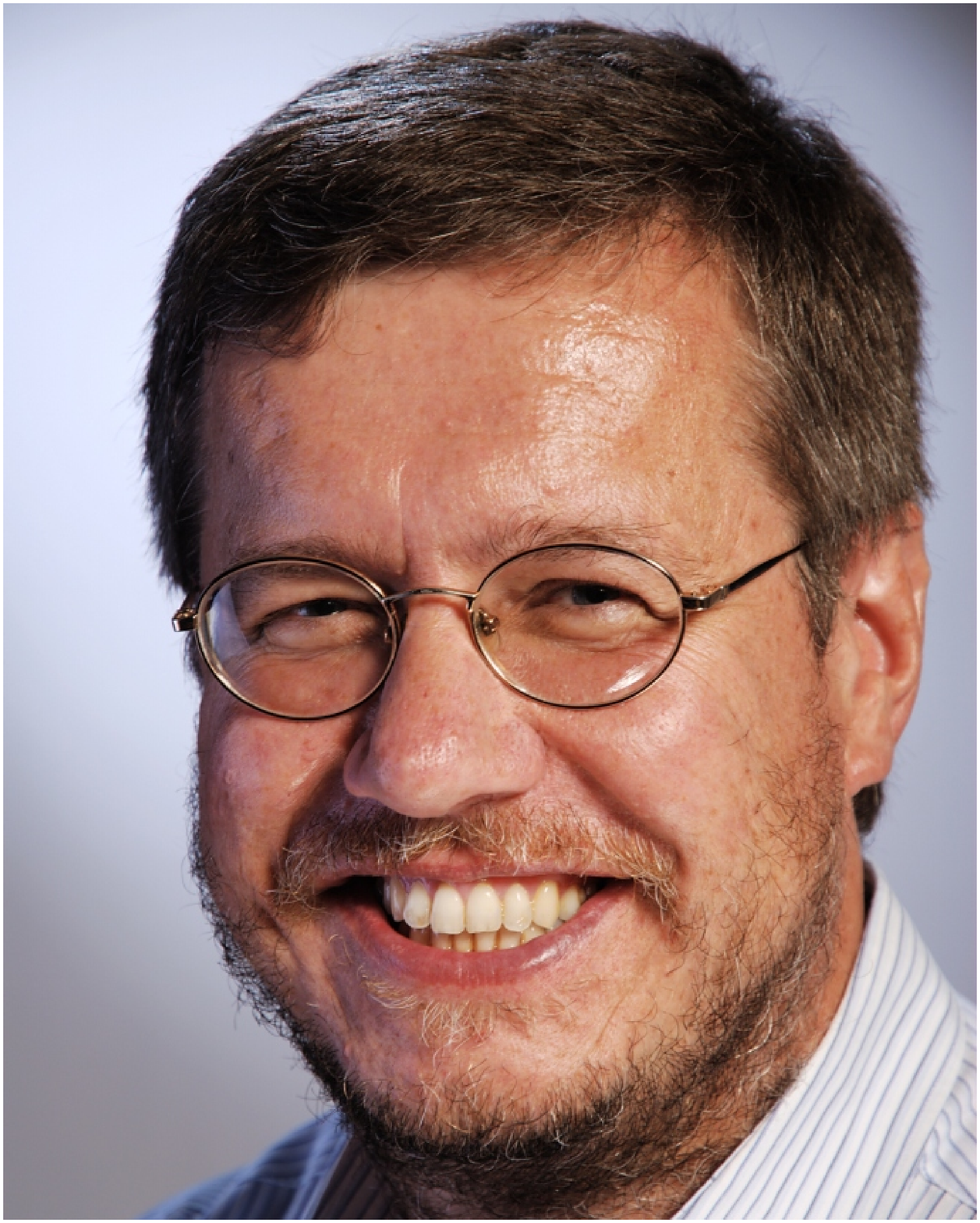}}]{Gernot Kubin}
(M'84) was born in Vienna, Austria, on June 24, 1960. He received his Dipl.-Ing. (1982) and Dr.techn. (1990, sub auspiciis praesidentis) degrees in Electrical Engineering from TU Vienna. He is Professor of Nonlinear Signal Processing and head of the Signal Processing and Speech Communication Laboratory (SPSC) at TU Graz/Austria since 2000. At TU Graz, he has been Dean of Studies in EE-Audio Engineering 2004-2007, Chair of the Senate 2007-2010 and 2013-now, and he has coordinated the Doctoral School in Information and Communications Engineering since 2007. Earlier international appointments include: CERN Geneva/CH 1980, TU Vienna 1983-2000, Erwin Schroedinger Fellow at Philips Natuurkundig Laboratorium Eindhoven/NL 1985, AT\&T Bell Labs Murray Hill/USA 1992-1993 and 1995, KTH Stockholm/S 1998, and Global IP Sound Sweden \& USA 2000-2001 and 2006, UC San Diego \& UC Berkeley/USA 2006, and UT Danang, Vietnam 2009. In 2011, he has co-founded Synvo GmbH, a start-up in the area of speech synthesis for mobile devices, and he holds leading positions in several national research centres for academia-industry collaboration such as the Vienna Telecommunications Research Centre FTW 1999-now (Key Researcher and Board of Governors), the Christian Doppler Laboratory for Nonlinear Signal Processing 2002-2010 (Founding Director), the Competence Network for Advanced Speech Technologies COAST 2006-2010 (Scientific Director), the COMET Excellence Projects Advanced Audio Processing AAP 2008-2013 and Acoustic Sensing and Design 2013-now (Key Researcher), and in the National Research Network on Signal and Information Processing in Science and Engineering SISE 2008-2011 (Principal Investigator) funded by the Austrian Science Fund. Since 2011, Dr.Kubin is an elected member of the Speech and Language Processing Technical Committee of the IEEE. His research interests are in nonlinear signals and systems, computational intelligence, as well as speech and audio communication. He has authored or co-authored over one hundred fifty peer-reviewed publications and ten patents.
\end{IEEEbiography}

\begin{IEEEbiography}[{\includegraphics[width=1in,height=1.25in,clip,keepaspectratio]{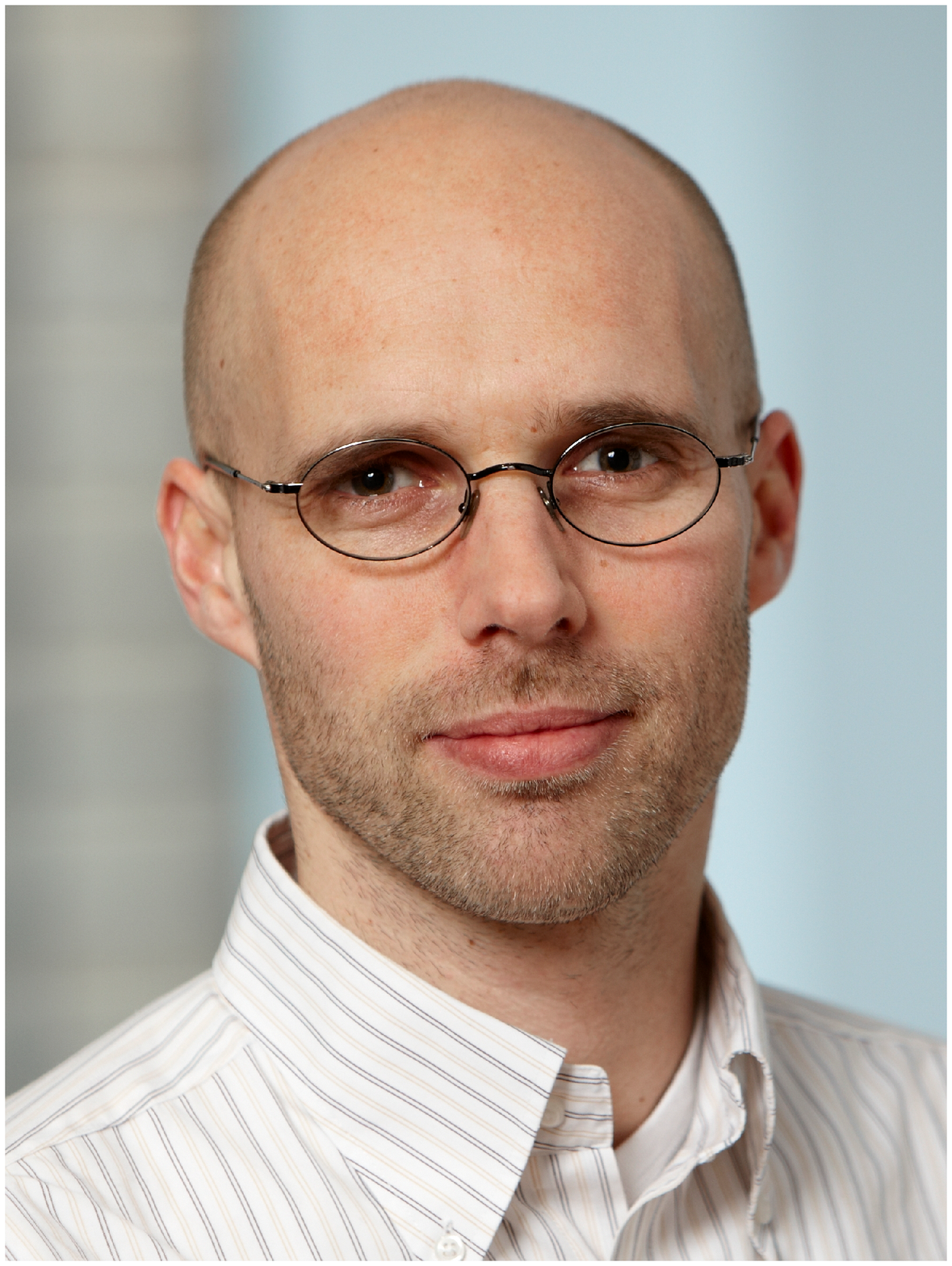}}]{Heinz Koeppl}
received the M.Sc. degree in physics from Graz Karl-Franzens University, Austria in 2001 and the Ph.D. degree in electrical engineering from Graz University of Technology, Austria in 2004. He has been a postdoctoral fellow at UC Berkeley, USA and at EPF Lausanne Switzerland until 2009 and an assistant professor at ETH Zurich, Switzerland until 2013. He is currently a full professor of Electrical Engineering at Technische Universitaet Darmstadt, Germany. He received the Erwin Schr\"odinger fellowship (2005) the IFAC Fred Margulies Ph.D. thesis award (2006), the CNRS/Max-Planck postdoctoral fellowship (2008), the SNSF Professorship (2010) and the IBM Faculty Award (2014). His research interests include the analysis and reconstruction of stochastic biomolecular networks and self-organizing phenomena such as flocking and self-assembly.
\end{IEEEbiography}

\end{document}